\documentclass[journal]{IEEEtran}

\usepackage[T1]{fontenc}
\usepackage[latin9]{inputenc}
\usepackage{float}
\usepackage{amsmath}
\usepackage{amsthm}
\usepackage{amssymb}
\usepackage{graphicx}
\usepackage{stfloats}
\usepackage{cite} 
\usepackage{xcolor}
\usepackage{bbm}
\usepackage{multirow}
\usepackage{multicol}
\usepackage{subfig}

\usepackage{amsmath,amssymb,amsfonts}
\usepackage{algorithmic}
\usepackage{graphicx}
\usepackage{textcomp}
\usepackage{algorithm}
\usepackage{comment}
\usepackage{nomencl}
\usepackage{multirow}
\usepackage{multicol}
\usepackage{amsthm}
\usepackage{xcolor}
\makenomenclature

\allowdisplaybreaks

\usepackage{bm}

\theoremstyle{theorem}
\newtheorem{theorem}{Theorem} 

\makeatletter

\floatstyle{ruled}
\newfloat{algorithm}{tbp}{loa}
\providecommand{\algorithmname}{Algorithm}
\floatname{algorithm}{\protect\algorithmname}

\theoremstyle{plain}

\theoremstyle{plain}

\IEEEoverridecommandlockouts
\usepackage{cite}
\usepackage{amsfonts}\usepackage{algorithmic}
\usepackage{textcomp}
\usepackage{etoolbox}

\def\BibTeX{{\rm B\kern-.05em{\sc i\kern-.025em b}\kern-.08em
    T\kern-.1667em\lower.7ex\hbox{E}\kern-.125emX}}

\patchcmd{\@maketitle}
  {\addvspace{0.5\baselineskip}\egroup}
  {\addvspace{-1\baselineskip}\egroup}
  {}
  {}

\makeatother

\providecommand{\propositionname}{Proposition}
\providecommand{\theoremname}{Theorem}

\begin{document}

\title{\huge Sequential In-Network Processing for Cell-Free Massive MIMO with Capacity-Constrained Parallel Radio Stripes}

\author{Sangwon Jo, \textit{Graduate Student Member}, \textit{IEEE}, Hoon Lee, \textit{Member}, \textit{IEEE}, \\and Seok-Hwan Park, \textit{Senior Member}, \textit{IEEE} \thanks{


This work was supported in part by the National Research Foundation (NRF) of
Korea, funded by the MOE under Grant 2019R1A6A1A09031717 and MSIT under Grant RS-2023-00238977, in part by research funds of Jeonbuk National University in 2024, and in part by a Korea Research Institute
for Defense Technology Planning and Advancement
grant funded by the Korea government (Defense Acquisition
Program Administration) (21-106-A00-007, Space-
Layer Intelligent Communication Network Laboratory,
2022).  \textit{(Corresponding authors: Hoon Lee and Seok-Hwan Park.)}

S. Jo and S.-H. Park are with the Division of Electronic Engineering, Jeonbuk
National University, Jeonju, South Korea (email: tkddnjs9803@jbnu.ac.kr, seokhwan@jbnu.ac.kr).

H. Lee is with the Department of Electrical Engineering, Ulsan National Institute of Science and Technology (UNIST), Ulsan 44919, South Korea (email: hoonlee@unist.ac.kr).

Copyright (c) 2024 IEEE. Personal use of this material is permitted. However, permission to use this material for any other purposes must be obtained from the IEEE by sending a request to pubs-permissions@ieee.org.}
}
\maketitle
\begin{abstract}
To ensure coherent signal processing across distributed Access Points (APs) in Cell-Free Massive Multiple-Input Multiple-Output (CF-mMIMO) systems, 
a fronthaul connection between the APs and a Central Processor (CP) is imperative.
We consider a fronthaul network employing parallel radio stripes. In this system, APs are grouped into multiple segments where APs within each segment are sequentially connected through a radio stripe. This fronthaul topology strikes a balance between standard star and bus topologies, which deploy parallel or serial connections of all APs.
Our focus lies in designing the uplink signal processing for a CF-mMIMO system with parallel radio stripes. We tackle the challenge of finite-capacity fronthaul links by addressing the design of In-Network Processing (INP) strategies at APs. These strategies involve linearly combining received signals and compressing the combining output for fronthaul transmission, aiming to maximize the sum-rate performance.
Given the high complexity and the stringent requirement for global Channel State Information (CSI) in jointly optimizing INP strategies across all APs, we propose an efficient sequential design approach. 
Numerical results demonstrate that the proposed sequential INP design achieves a sum-rate gain of up to $82.92$\% compared to baseline schemes.

\end{abstract}


\section{Introduction}
\label{sec:introduction}

\subsection{Background and Motivation} \label{sub:background-motivation}

Cell-Free Massive Multiple-Input Multiple-Output (CF-mMIMO) stands out as a promising wireless networking technology facilitating pervasive connectivity across extensive coverage areas \cite{Bjornson-et-al:TWC20, Ngo:proc24, Zhang:Acc19, Kim:TVT22, Yu:TVT23}. In this architecture, Access Points (APs) distributed across the coverage are coherently managed by a Central Processor (CP) via a fronthaul network. 
In the uplink scenario, the signals transmitted by mobile User Equipments (UEs) are jointly decoded at the CP, leveraging the global observation gathered from the fronthaul links. This coordinated approach results in significantly higher spectral and energy efficiency compared to traditional cellular networks, including conventional mMIMO systems \cite{Molish:CM17, Sharifi:Acc23} which rely on independent mMIMO processing of each Base Station (BS).
Most of the existing literature, including \cite{Zhang:Acc19, Kim:TVT22, Yu:TVT23}, has concentrated on a star fronthaul topology, in which each AP is equipped with a dedicated fronthaul link to the CP. However, this approach necessitates the utilization of expensive cables to establish individual fronthaul connections for each AP.

To mitigate this concern, a bus topology, also known as a radio stripe fronthaul network, has recently been explored \cite{Interdonato:EJWCN19, Shaik-et-al:TC21, Ma-et-al:ICC21,Chiotis-et-al:Meditcom22, Jo:GC23}. By sequentially connecting the APs to the CP with a single cable, the radio stripe fronthaul network facilitates the cost-effective deployment of CF-mMIMO systems. 
Despite the cost-effective fronthaul deployment, the radio stripe introduces additional fronthaul latency compared to the star topology to propagate information signals over the entire APs.
To achieve a balance between the star topology and the radio stripe, one can consider a segmented fronthaul network approach \cite{Fernandes:WCL22}. In this setup, APs are partitioned into multiple segments, with APs within each segment being sequentially connected through a radio stripe or a segment bus.
Consequently, the segmented fronthaul network can be viewed as comprising parallel radio stripes.

To the best of our knowledge, while there have been active research efforts focused on signal processing design in CF-mMIMO systems as introduced in Sec. \ref{sub:related-work}, the development of efficient algorithm for the uplink of CF-mMIMO system with capacity-constrained parallel radio stripes remains unaddressed in the existing literature. 
This gap motivated us to tackle this open problem, as summarized in Sec. \ref{sub:main-contribution}.

\begin{table*}[t]
\caption{\small \centering Comparison of related works.}
\centering
\renewcommand{\arraystretch}{1.5}
\begin{tabular}{|c||c||c|c||c|}
    \hline
    \multirow{2}{*}{Papers} & \multirow{2}{*}{Uplink/downlink} & \multicolumn{2}{c||}{Fronthaul} & \multirow{2}{*}{Performance metric} \\
    \cline{3-4}
    & & Topology & Compression &  \\
    \hline \hline
    \multirow{2}{*}{\cite{Bjornson-et-al:TWC20}} & \multirow{2}{*}{Uplink} & Star topology with an & \multirow{2}{*}{Not considered} & \multirow{2}{*}{Per-UE rate} \\
    & & example for a single radio stripe &  & \\
    \hline
    \cite{Kim:TVT22} & Downlink & Star topology & Considered & Max-min fairness \\
    \hline
    \cite{Yu:TVT23} & Downlink & Star topology & Not considered & Sum-rate \\
    \hline
    \cite{Shaik-et-al:TC21} & Uplink & Single radio stripe & Not considered & Per-UE rate \\
    \hline
    \cite{Ma-et-al:ICC21} & Uplink & Single radio stripe & Not considered & Per-UE rate \\
    \hline
    \cite{Chiotis-et-al:Meditcom22} & Uplink & Single radio stripe & Considered & Per-UE rate \\
    \hline
    \cite{Jo:GC23} & Downlink & Single radio stripe & Considered & Sum-rate \\
    \hline
    \cite{Fernandes:WCL22} & Downlink & Parallel radio stripes & Not considered & Per-UE rate \\
    \hline
    This work & Uplink & Parallel radio stripes & Considered & Sum-rate \\
    \hline
\end{tabular}
~

~

\end{table*}

\subsection{Related Work} \label{sub:related-work}

Early studies on CF-mMIMO systems, such as \cite{Bjornson-et-al:TWC20, Kim:TVT22, Yu:TVT23}, mostly focused on \textit{star topology} where each AP is equipped with individual dedicated fronthaul links, either with explicit capacity constraints  \cite{Kim:TVT22} or without \cite{Bjornson-et-al:TWC20, Yu:TVT23}.
In particular, \cite{Bjornson-et-al:TWC20} explored various uplink reception techniques based on different levels of fronthaul cooperation, while \cite{Kim:TVT22} and \cite{Yu:TVT23} proposed downlink power control algorithms, leveraging alternating optimization \cite{Kim:TVT22} and deep learning-based fronthaul coordination \cite{Yu:TVT23}, respectively.
Notably, CF-mMIMO setups with a star fronthaul topology exhibit significant similarities to the Cloud Radio Access Network (C-RAN) and Fog Radio Access Network (F-RAN) architectures that have been extensively studied (see, e.g., \cite{Wu:Net15, Peng:Net16, Lee:TWC21}).

The CF-mMIMO system employing a \textit{single radio stripe} was examined in \cite{Shaik-et-al:TC21, Ma-et-al:ICC21, Chiotis-et-al:Meditcom22, Jo:GC23}. Due to the serial connection facilitated by a single fronthaul cable, the works in \cite{Shaik-et-al:TC21, Ma-et-al:ICC21, Chiotis-et-al:Meditcom22} focused on successive uplink signal processing across APs. Notably, \cite{Shaik-et-al:TC21} demonstrated that, in the absence of fronthaul capacity constraints, carefully-designed successive local processing with partial information exchange among neighboring APs can achieve performance on par with the optimal centralized reception scheme.
On the downlink side, \cite{Jo:GC23} explored joint optimization of downlink precoding and fronthaul compression. Leveraging the fact that each AP has access to all the compressed bit blocks, which pass though it, due to the serial nature of the data transfer, the application of an advanced Wyner-Ziv (WZ) compression technique \cite{Xiong:SPM04} was proposed.

The spectral efficiency of CF-mMIMO system using \textit{parallel radio stripes} was investigated in \cite{Fernandes:WCL22}, focusing on the downlink scenario.
Fronthaul imperfections were modeled via binary random variables to represent random disconnections for individual fronthaul links. However, \cite{Fernandes:WCL22} primarily analyzed the impact of parallel radio stripe deployment in the downlink setting, without addressing explicit fronthaul capacity limitations. 
As a result, the signal processing design for uplink CF-mMIMO systems with capacity-constrained parallel radio stripes remains an open research problem.

Table 1 provides a summary of the above-mentioned related studies, detailing the scenarios considered (uplink/downlink, fronthaul topology, and assumptions on fronthaul capacity) along with the evaluated performance metrics. This highlights the novelty of our work compared to existing literature.

\subsection{Main Contribution} \label{sub:main-contribution}

In this work, we explore the uplink signal processing of a CF-mMIMO system employing a segmented fronthaul network with parallel radio stripes. To the best of our knowledge, this is the first effort to address the signal processing design of CF-mMIMO systems with parallel radio stripes. Building on prior works that examined practical fronthaul connections in the star \cite{Yu:TVT23} and bus topologies \cite{Jo:GC23}, we assume that digital links with finite capacity are deployed on each radio stripe cable. Within this framework, each AP executes In-Network Processing (INP) \cite{Park:TVT16}, wherein baseband signals received from the uplink wireless channel and the incoming fronthaul link are linearly combined and compressed to generate a bit stream for transmission over the outgoing fronthaul link.
Compared to \cite{Shaik-et-al:TC21}, our system investigates a more practical CF-mMIMO with multiple radio stripes each with finite capacity. This incurs additional optimization on fronthaul quantization strategies together with signal processing over parallel radio stripes.

Our objective is to maximize the sum-rate of the CF-mMIMO network by optimizing the INP strategies, i.e., the linear combining and fronthaul compression strategies. Due to the high complexity and the stringent requirement for global Channel State Information (CSI) associated with jointly optimizing the INP strategies across all APs, we propose an efficient sequential design approach on a per-AP basis. For each AP, we first determine the linear combining matrices adopting the Minimum Mean Squared Error (MMSE) criterion.
It is followed by the optimization of the fronthaul quantization noise powers.
Despite the non-convex nature of the problem, we derive a closed-form solution that ensures a stationary point of the problem.
The proposed INP optimization can be realized with a sequential fronthaul coordination where neighboring APs exchange Side Information (SI) only.
Through numerical evaluations, we validate the performance gains of the proposed sequential design against baseline schemes.

Main contributions of this work are listed below.
\begin{itemize}
    \item \textbf{Novel INP Design:} This work presents the first investigation into the INP design for the uplink of CF-mMIMO systems equipped with parallel radio stripes. To ensure the practical applicability of our approach, we consider both capacity-constrained radio stripes and the effects of imperfect CSI.
    \item \textbf{Efficient Sequential Algorithm:} We propose a computationally efficient sequential INP design algorithm that determines the linear combining and fronthaul compression strategies on a per-AP basis. At each AP, the linear combining matrix is optimized based on the MMSE criterion, while the fronthaul quantization noise powers are adjusted to maximize the achievable sum-rate under parallel decoding using the quantizer outputs.
    \item \textbf{Sequential Implementation:} We outline the sequential implementation of the proposed INP optimization, detailing the minimal SI that needs to be exchanged between neighboring APs.
\end{itemize}

\subsection{Organization} \label{sub:organization}

The remainder of this paper is structured as follows. Sec. \ref{sec:System-Model} introduces the system model for the uplink of a CF-mMIMO system employing parallel radio stripes and describes the imperfect CSI model.
The sequential INP operations, involving linear combining and fronthaul compression, are detailed in Sec. \ref{sec:seq-INP}, followed by a discussion of the proposed INP optimization in Sec. \ref{sec:sequential-optimization} highlighting its sequential implementation.
Numerical results demonstrating the effectiveness of the proposed scheme are presented in Sec. \ref{sec:numerical-results}, and the conclusions are drawn in Sec. \ref{sec:conclusion}.

\section{System Models} \label{sec:System-Model}

\begin{figure}
\centering\includegraphics[width=0.8\linewidth]{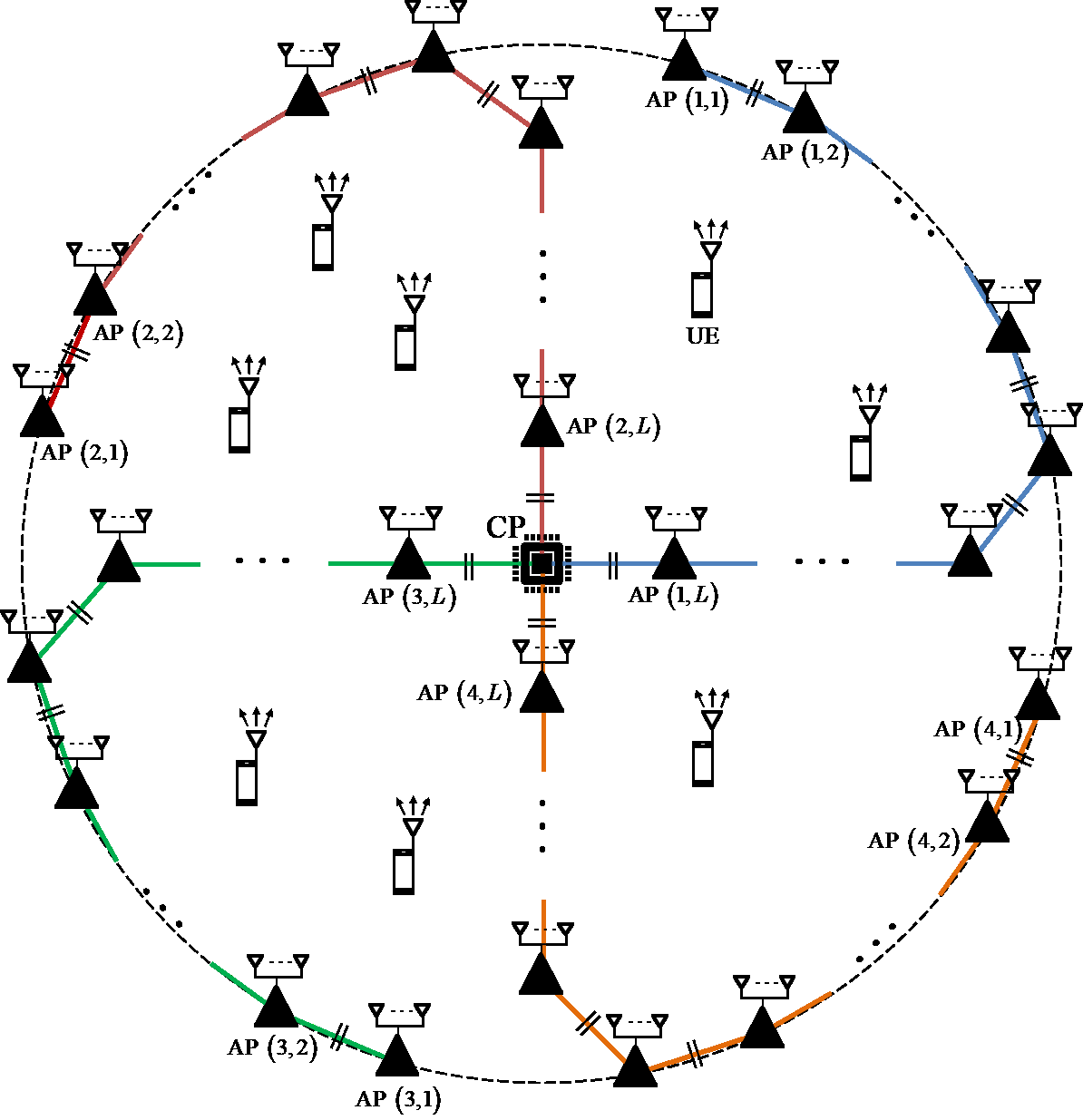}
\caption{\label{fig:system-model}Illustration of the uplink of a CF-mMIMO system with a fronthaul network of $M=4$ radio stripes.}
\end{figure}

We consider the uplink of a CF-mMIMO system employing parallel radio stripes \cite{Fernandes:WCL22}, depicted in Fig. \ref{fig:system-model}. 
There exist $K$ UEs which send their data signals to APs in the uplink. The APs are equipped with $N$ receive antennas, whereas the UEs are assumed to have single antenna\footnote{In this work, we focus on the proposed INP design at the APs, considering single-antenna UEs. The co-design with transmit covariance matrices for multi-antenna UEs is left for future research.}. The APs convey their received signals to a CP through fronthaul links. Then, with the fronthaul-received signals, the CP recovers transmitted signals of UEs by means of the joint signal processing.

\subsection{Fronthaul network model}
In the considered system, the fronthaul network comprises $M$ independent radio stripes each connecting $L$ APs and a CP. 
Index sets of radio stripe, UEs, and APs on each stripe are denoted by $\mathcal{M}=\{1,2,\ldots,M\}$, $\mathcal{K} = \{1,2,\ldots,K\}$, and $\mathcal{L} = \{1,2,\ldots,L\}$, respectively. The $i$th AP on the $m$th stripe is denoted by AP $(m,i)\in\mathcal{M}\times\mathcal{L}$. 
For convenience, the CP is regarded as AP $(m,L+1)$ for any given $m\in\mathcal{M}$.

APs on the $m$th stripe are linked to the CP through a sequential fronthaul connection with capacity of $C_F$ bps/Hz, traversing from AP $(m,1)$ through $(m,2)$, $(m,3)$, $\ldots$, to $(m,L)$, and then to the CP $(m,L+1)$. That is, to collect the data signals received by the APs, the information of APs flows the fronthaul link from AP $(m,i)$ to AP $(m,i+1)$, $\forall i\in\mathcal{L}$, thereby reaching the CP $(m,L+1)$. As a result, the parallel radio stripes can be modeled as a directed graph $\mathcal{G}=(\mathcal{V},\mathcal{E})$ where $\mathcal{V}$ stands for the vertex set collecting the CP and APs and $\mathcal{E}$ denotes the edge set consisting of the fronthaul links. The vertex and edge sets are respectively defined as
\begin{subequations}
\begin{align}
    \mathcal{V}&=\big\{(m,i)\big\}_{m\in\mathcal{M},i\in\mathcal{L}\bigcup\{L+1\}},  \nonumber \\
    \mathcal{E}&=\big\{[(m,i),(m,i+1)]\big\}_{m\in\mathcal{M},i\in\mathcal{L}}, \nonumber
\end{align}
\end{subequations}
where the directed edge $[(m,i),(m,i+1)]$ represents the fronthaul link from AP $(m,i)$ to $(m,i+1)$, $\forall i\in\mathcal{L}$.

\subsection{Channel Training and Estimation Error Model}
Let $\mathbf{h}_{m,i,k}\sim\mathcal{CN}(\mathbf{0}, \mathbf{R}_{m,i,k})\in\mathbb{C}^{N\times 1}$ respresent the channel vector from UE $k$ to AP $(m,i)$ with the spatial covariance matrix $\mathbf{R}_{m,i,k} = \mathbb{E}[\mathbf{h}_{m,i,k}\mathbf{h}_{m,i,k}^H]$.
The corresponding large-scale fading coefficient $\beta_{m,i,k}$, which accounts for the path-loss and shadowing, is given by $\beta_{m,i,k} = \text{tr}(\mathbf{R}_{m,i,k}) / N$.
In the remainder of this subsection, we describe the process of estimating the channel vectors $\{\mathbf{h}_{m,i,k}\}_{k\in\mathcal{K}}$ at AP $(m,i)$ during the channel training phase.

For the channel training, UE $k$ sends the pilot sequence $\boldsymbol{\phi}_{k}\in\mathbb{C}^{K}$ of length $K$ to the APs. It is assumed that $K$ pilot sequences $\{\boldsymbol{\phi}_k\in\mathbb{C}^{K}\}_{k\in\mathcal{K}}$ are orthogonal, i.e., $\boldsymbol{\phi}_k^H \boldsymbol{\phi}_l = K \delta_{k,l}$.
After pilot transmission over $K$ channel uses, the received pilot signal at AP $(m,i)$, denoted by $\mathbf{Y}_{m,i}^p \in \mathbb{C}^{N \times K}$, is given by
\begin{align}
    \mathbf{Y}_{m,i}^p = \sum\nolimits_{k\in\mathcal{K}} \sqrt{P_k} \mathbf{h}_{m,i,k} \boldsymbol{\phi}_k^T + \mathbf{Z}_{m,i}^p, \nonumber
\end{align}
where $P_k$ denotes the pilot transmit power of UE $k$ and $\mathbf{Z}_{m,i}^p$ represents the additive noise signal with elements identically and independently distributed as $\mathcal{CN}(0, \sigma_z^2)$.
To estimate the channel vector $\mathbf{h}_{m,i,k}$, AP $(m,i)$ despreads the received signal using the pilot sequence $\boldsymbol{\phi}_k$ and applies the linear MMSE estimator.
The corresponding estimated CSI vector $\hat{\mathbf{h}}_{m,i,k}\in\mathbb{C}^{N\times 1}$ is expressed as \cite{Shaik-et-al:TC21}
\begin{align}
    \hat{\mathbf{h}}_{m,i,k} = \sqrt{P_k K} \mathbf{R}_{m,i,k} \left( P_k K  \mathbf{R}_{m,i,k} + \sigma_z^2\mathbf{I} \right)^{-1} \mathbf{Y}_{m,i}^p \frac{\boldsymbol{\phi}_k^*}{\sqrt{K}}, \nonumber
\end{align}
which follows the distribution $\mathcal{CN}(\mathbf{0}, \hat{\mathbf{R}}_{m,i,k})$ with the covariance matrix $\hat{\mathbf{R}}_{m,i,k}$ given by \cite{Shaik-et-al:TC21}
\begin{align}
    \hat{\mathbf{R}}_{m,i,k} = P_{k}K\mathbf{R}_{m,i,k}(P_{k}K\mathbf{R}_{m,i,k} + \sigma^{2}_{z}\mathbf{I})^{-1}\mathbf{R}_{m,i,k}. \nonumber
\end{align}
It is well known that with the MMSE estimator, the estimated CSI $\hat{\mathbf{h}}_{m,i,k}$ and the estimation error vector $\tilde{\mathbf{h}}_{m,i,k} = \mathbf{h}_{m,i,k} - \hat{\mathbf{h}}_{m,i,k} \sim \mathcal{CN}(\mathbf{0}, \tilde{\mathbf{R}}_{m,i,k})$ are uncorrelated, such that $\mathbf{R}_{m,i,k} = \hat{\mathbf{R}}_{m,i,k} +  \tilde{\mathbf{R}}_{m,i,k}$.

Through the channel estimation procedure, AP $(m,i)$ can know its local CSI estimate $\hat{\mathbf{H}}_{m,i} = [\hat{\mathbf{h}}_{m,i,1} \, \cdots \, \hat{\mathbf{h}}_{m,i,K}]$ and the associated error covariance matrices $\tilde{\mathbf{R}}_{m,i,k}$, $\forall k\in\mathcal{K}$. On the contrary, the local CSI statistics of other APs $\hat{\mathbf{H}}_{n,j}$ and $\tilde{\mathbf{R}}_{n,j,k}$, $\forall n\neq m, j\neq i$, are, in general, unavailable at AP $(m,i)$. In practice, the fronthaul coordination is necessary to share such locally available information with neighboring APs, which might incur additional fronthaul coordination overhead.

\subsection{Payload Data Transmission on Access Channel}
Let $x_{k}\sim\mathcal{CN}(0,P_{k})$ be the uplink data symbol transmitted at UE $k$ during the payload data transmission phase. We define the transmitted signal vector of all UEs as $\mathbf{x} = [x_1 , \cdots , x_K]^T \sim \mathcal{CN}(\mathbf{0}, \boldsymbol{\Sigma}_{\mathbf{x}})$, where the covariance matrix $\boldsymbol{\Sigma}_{\mathbf{x}}$ is given as  $\boldsymbol{\Sigma}_{\mathbf{x}} = \text{diag}(\{P_k\}_{k\in\mathcal{K}})$.
With the aforementioned CSI model, the received signal vector $\mathbf{y}_{m,i}\in\mathbb{C}^{N\times 1}$ at AP $(m,i)$ can be expressed as
\begin{align}
\mathbf{y}_{m,i} &= \hat{\mathbf{H}}_{m,i}\mathbf{x} + \tilde{\mathbf{H}}_{m,i}\mathbf{x} + \mathbf{z}_{m,i} \nonumber \\
&= \hat{\mathbf{H}}_{m,i}\mathbf{x} + \mathbf{w}_{m,i}, \nonumber
\end{align}
where $\hat{\mathbf{H}}_{m,i} = [\hat{\mathbf{h}}_{m,i,1} \, \cdots \, \hat{\mathbf{h}}_{m,i,K}]$ stands for the stacked estimated CSI from all UEs to AP $(m,i)$, $\tilde{\mathbf{H}}_{m,i} = [\tilde{\mathbf{h}}_{m,i,1} \, \cdots \, \tilde{\mathbf{h}}_{m,i,K}]$ is the corresponding CSI estimation error matrix, and $\mathbf{z}_{m,i}\sim\mathcal{CN}(\mathbf{0}, \sigma_z^2\mathbf{I})$ is the additive Gaussian noise vector. Here, $\mathbf{w}_{m,i} = \tilde{\mathbf{H}}_{m,i}\mathbf{x} + \mathbf{z}_{m,i}$ represents the effective noise which incorporates the erroneous CSI and Gaussian noise, whose covariance matrix $\boldsymbol{\Sigma}_{\mathbf{w}_{m,i}}\in\mathbb{C}^{N\times N}$ is given by
\begin{align}
    \boldsymbol{\Sigma}_{\mathbf{w}_{m,i}} =\mathbb{E}[\mathbf{w}_{m,i}\mathbf{w}_{m,i}^{H}]= \sum\nolimits_{k\in\mathcal{K}} P_k \tilde{\mathbf{R}}_{m,i,k} + \sigma_z^2\mathbf{I}. \label{eq:covariance-Rx-signal}
\end{align}

Due to the mutual interference among the UEs and APs, the centralized signal processing at the CP is an essential requirement to recover the uplink signal vector $\mathbf{x}$ reliably. To this end, the CP needs to know both the estimated CSI $\hat{\mathbf{H}}_{m,i}$ and received signal $\mathbf{y}_{m,i}$, which are only available at AP $(m,i)$. We can exploit the fronthaul network to forward the AP-side information $\hat{\mathbf{H}}_{m,i}$ and $\mathbf{y}_{m,i}$ to the CP. However, unlike the ideal star topology where all APs are directly connected to the CP, the sequential interaction nature of the parallel radio stripes incurs an extensive latency in collecting the AP-side information sequentially. In addition, finite-capacity fronthaul links do not allow the perfect sharing of the information. Moreover, handling all the received signals $\mathbf{y}_{m,i}$, $\forall m\in\mathcal{M},i\in\mathcal{L}$, solely at the CP incurs a huge complexity burden.

To address these issues, in this paper, we propose an INP strategy where the signal processing calculations of the received vectors $\mathbf{y}_{m,i}$, $\forall m\in\mathcal{M},i\in\mathcal{L}$, are executed by individual APs along with the information transmission over the parallel radio stripes. The proposed approach does not require the dedicated CSI exchange among the APs and centralized signal processing at the CP, thereby reducing the overall latency in the fronthaul coordination and signal processing. The major features of the proposed INP method are two-fold: sequential in-network signal processing and sequential optimization. The detailed operations of the proposed approach are presented in the following sections.

\section{Sequential In-Network Signal Processing} \label{sec:seq-INP}

\begin{figure*}[t]
\centering\includegraphics[width=0.7\linewidth]{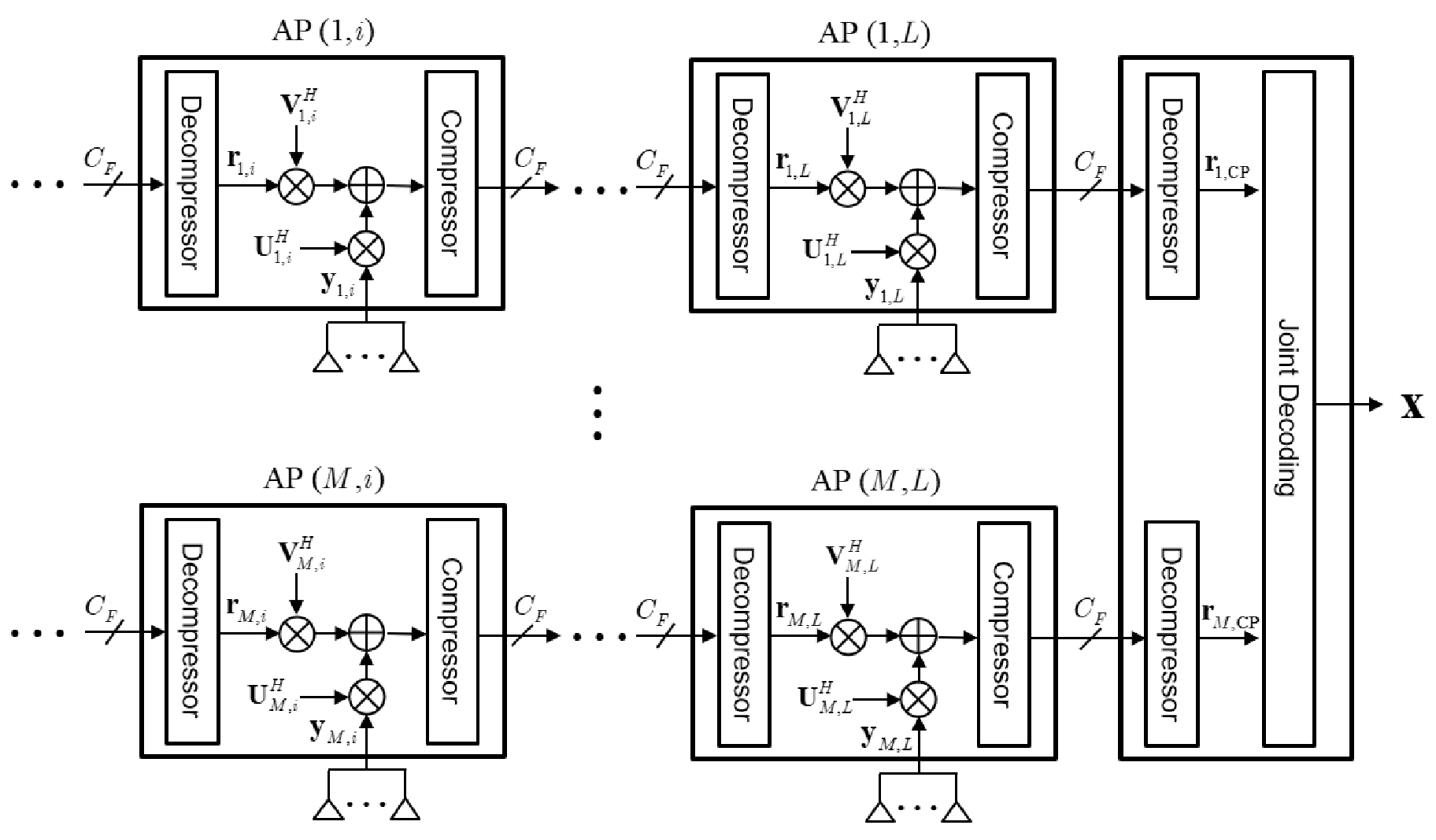}
\caption{\label{fig:INP}Illustration of the INP operations at the APs.}
\end{figure*}

In this section, we first describe the proposed INP method which facilitates the parallel radio stripe fronthaul topology through sequential signal processing among the APs. It is then followed the problem formulation for optimizing key components of the proposed INP.

\subsection{Proposed Sequential INP Strategy}
As illustrated in Fig.~\ref{fig:INP}, the proposed INP operations involve linear combining and fronthaul compression at individual APs. AP $(m,i)$ obtains two different signals $\mathbf{y}_{m,i}$ and $\mathbf{r}_{m,i-1}\in\mathbb{C}^{K\times1}$, which are received from the UEs in the uplink wireless access channels and from preceding AP $(m,i-1)$ in the incoming fronthaul link, respectively. After the INP, AP $(m,i)$ produces the output signal $\mathbf{r}_{m,i}$, which is shared to AP $(m,i+1)$. For convenience, we set $\mathbf{r}_{m,0} = \mathbf{0}$ for AP $(m,1)$, $\forall m\in\mathcal{M}$.
Note that since the CP is responsible for detecting the signal vector $\mathbf{x}$ transmitted by all UEs, we set the dimension of the INP output signal $\mathbf{r}_{m,i}$ to $K$ which corresponds to the number of UEs.
This choice has the advantage that the dimension of $\mathbf{r}_{m,i}$ remains independent of the other system parameters determining the network size, such as $M$, $L$, and $N$.

To generate the signal $\mathbf{r}_{m,i}$ conveyed to subsequent AP $(m,i+1)$, AP $(m,i)$ first performs linear combining of $\mathbf{y}_{m,i}$ and $\mathbf{r}_{m,i-1}$. The resulting output $\tilde{\mathbf{r}}_{m,i}\in\mathbb{C}^{K\times1}$ is obtained as
\begin{align}
\tilde{\mathbf{r}}_{m,i}=\mathbf{U}_{m,i}^H\mathbf{y}_{m,i} + \mathbf{V}_{m,i}^H\mathbf{r}_{m,i-1}, \label{eq:r_tilde}
\end{align} 
where $\mathbf{U}_{m,i}\in\mathbb{C}^{N\times K}$ and $\mathbf{V}_{m,i}\in\mathbb{C}^{K\times K}$ represent the combining matrices for $\mathbf{y}_{m,i}$ and $\mathbf{r}_{m,i-1}$, respectively. Subsequently, AP $(m,i)$ executes fronthaul compression to produce a compressed bit stream representing a quantized version $\mathbf{r}_{m,i}$ of the combining output signal $\tilde{\mathbf{r}}_{m,i}$, which is then sent to AP $(m,i+1)$, $\forall i\in\mathcal{L}$. 
Assuming a Gaussian test channel \cite{Park:TVT16, Zhou:TSP16, Park:TVT22} for the compression, the decompressed signal $\mathbf{r}_{m,i}\in\mathbb{C}^{K\times 1}$ at AP $(m,i+1)$ is expressed as
\begin{align}
\mathbf{r}_{m,i} = \tilde{\mathbf{r}}_{m,i}+\mathbf{q}_{m,i}, \label{eq:quantized-signal-AP-m-i}
\end{align}
where $\mathbf{q}_{m,i}\sim\mathcal{CN}(\mathbf{0},\boldsymbol{\Omega}_{m,i})$ is the quantization noise with the covariance matrix $\boldsymbol{\Omega}_{m,i}\in\mathbb{C}^{K\times K}$. Notice that the linear combining \eqref{eq:r_tilde} and compression \eqref{eq:quantized-signal-AP-m-i} are carried out sequentially from AP $(m,1)$ to AP $(m,L)$, $\forall m\in\mathcal{M}$. Then, the CP recovers the UE signal vector $\mathbf{x}$ based on the fronthaul-received signals $\mathbf{r}_{m,L}$, $\forall m\in\mathcal{M}$.

According to the sequential processing, the signal $\mathbf{r}_{m,i}$ can be rewritten by
\begin{align}
\mathbf{r}_{m,i} = \hat{\mathbf{G}}_{m,i}\mathbf{x} + \mathbf{e}_{m,i}, \label{eq:rewritting-r}
\end{align}
where $\hat{\mathbf{G}}_{m,i}$ indicates the effective channel, and $\mathbf{e}_{m,i}$ is the noise component that collects the additive Gaussian noise, CSI error and quantization noise. These are respectively defined as
\begingroup
\allowdisplaybreaks
\begin{subequations} \label{eq:definition-Ghat-e}
\begin{align}
    &\hat{\mathbf{G}}_{m,i} = \sum\nolimits_{j=1}^i\mathbf{G}_{m,i,j}\hat{\mathbf{H}}_{m,j}, \label{eq:definition-Ghat} \\
    & \mathbf{e}_{m,i} = \sum\nolimits_{j=1}^{i} \mathbf{G}_{m,i,j}\mathbf{w}_{m,j} + \sum\nolimits_{j=1}^{i} \mathbf{T}_{m,i,j}\mathbf{q}_{m,j}, \label{eq:definition-e}
\end{align}    
\end{subequations}
\endgroup
where the matrices $\mathbf{G}_{m,i,j}\in\mathbb{C}^{K\times N}$ and $\mathbf{T}_{m,i,j}\in\mathbb{C}^{K\times K}$ are given by
\begin{subequations}
\begin{align}
    &\mathbf{G}_{m,i,j} = \mathbf{T}_{m,i,j} \mathbf{U}_{m,j}^H, \nonumber\\
    &\mathbf{T}_{m,i,j} = 
    \begin{cases}
        \prod\nolimits_{l=0}^{i-j-1} \!\mathbf{V}_{m,i-l}^H, & j\neq i,\\
        \mathbf{I}, & j=i.
    \end{cases} \nonumber
\end{align}
\end{subequations}
The covariance matrix of the noise component $\mathbf{e}_{m,i}$, denoted by $\boldsymbol{\Sigma}_{\mathbf{e}_{m,i}}$, is given by
\begin{align}
&\boldsymbol{\Sigma}_{\mathbf{e}_{m,i}} = \mathbb{E}[\mathbf{e}_{m,i}\mathbf{e}_{m,i}^H] \label{eq:covariance-e-i} \\
& = \sum\nolimits_{j=1}^i \mathbf{G}_{m,i,j}\boldsymbol{\Sigma}_{\mathbf{w}_{m,j}}\mathbf{G}_{m,i,j}^H + \sum\nolimits_{j=1}^i \mathbf{T}_{m,i,j}\boldsymbol{\Omega}_{m,j}\mathbf{T}_{m,i,j}^H. \nonumber
\end{align}
Based on rate-distortion theory, the signal $\mathbf{r}_{m,i}$ can be reliably recovered by AP $(m,i+1)$ if the mutual information between $\tilde{\mathbf{r}}_{m,i}$ and $\mathbf{r}_{m,i}$ is below the fronthaul capacity $C_{F}$, i.e., $I(\tilde{\mathbf{r}}_{m,i} ; \mathbf{r}_{m,i}) \leq C_F$ \cite{Park:TVT16, Zhou:TSP16, Park:TVT22}. 
This leads to a fronthaul capacity constraint expressed as
\begin{align}
f_{\text{FH},m,i}\left(\mathbf{U}, \mathbf{V}, \boldsymbol{\Omega}\right) &= \log_2\det\left( \hat{\mathbf{G}}_{m,i}\boldsymbol{\Sigma}_{\mathbf{x}}\hat{\mathbf{G}}_{m,i}^H + \boldsymbol{\Sigma}_{\mathbf{e}_{m,i}} \right) \nonumber \\
& - \log_2\det\left(\boldsymbol{\Omega}_{m,i}\right) \leq C_F, \label{eq:fronthaul-constraint-general}
\end{align}
where $\mathbf{U} = \{\mathbf{U}_{m,i}\}_{m\in\mathcal{M},i\in\mathcal{L}}$, $\mathbf{V} = \{\mathbf{V}_{m,i}\}_{m\in\mathcal{M},i\in\mathcal{L}}$, and $\boldsymbol{\Omega} = \{\boldsymbol{\Omega}_{m,i}\}_{m\in\mathcal{M},i\in\mathcal{L}}$.

\subsection{Problem Formulation}
Upon receiving the signals $\mathbf{r}_{m,L}$ from the last APs $(m,L)$ of all radio stripes $\forall m\in\mathcal{M}$, the CP decodes the uplink signal vector of UEs $\mathbf{x}$. The stacked fronthaul-received signal vector $\mathbf{r}_\text{CP} = [\mathbf{r}_{1,L}^H \, \cdots \, \mathbf{r}_{M,L}^H]^H \in \mathbb{C}^{K M \times 1}$ can be written by
\begin{align}
    \mathbf{r}_{\text{CP}} = \hat{\mathbf{G}}_{\text{CP}} \mathbf{x} + \mathbf{e}_{\text{CP}}, \nonumber
\end{align}
where $\hat{\mathbf{G}}_{\text{CP}} = [\hat{\mathbf{G}}_{1,L}^H \, \cdots \, \hat{\mathbf{G}}_{M,L}^H]^H \in \mathbb{C}^{K M \times K}$ collects the effective channel matrices and $\mathbf{e}_{\text{CP}} = [\mathbf{e}_{1,L}^H \, \cdots \, \mathbf{e}_{M,L}^H]^H\in\mathbb{C}^{K M \times 1}$ stands for the stacked noise component whose covariance matrix $\boldsymbol{\Sigma}_{\mathbf{e}_{\text{CP}}}$ is given by
\begin{align}
    \boldsymbol{\Sigma}_{\mathbf{e}_{\text{CP}}}=\mathbb{E}[\mathbf{e}_{\text{CP}}\mathbf{e}_{\text{CP}}^H] = \text{blkdiag}(\{\boldsymbol{\Sigma}_{\mathbf{e}_{m,L}}\}_{m\in\mathcal{M}}). \nonumber
\end{align}
Then, the sum-rate can be computed as \cite{Gamal:Cambridge11}
\begin{align}
    R_{\text{sum}}=\mathbb{E}_{\tilde{\mathbf{h}}}\left[I\left( \mathbf{x} ; \mathbf{r}_{\text{CP}} \right)\right], \nonumber
\end{align}
where the expectation is over the CSI error $\tilde{\mathbf{h}} = \{\tilde{\mathbf{h}}_{m,i,k}\}_{m\in\mathcal{M},i\in\mathcal{L},k\in\mathcal{K}}$.
Applying the Jensen's inequality with Gaussian approximation on the effective noise signal $\mathbf{e}_{\text{CP}}$ \cite{Choi:TWC20}, a lower bound on the sum-rate $R_{\text{sum}}$, denoted by $f_{\text{sum}}\left(\mathbf{U}, \mathbf{V}, \boldsymbol{\Omega}\right)$, is obtained as
\begin{align}
    f_{\text{sum}}\left(\mathbf{U}, \mathbf{V}, \boldsymbol{\Omega}\right) = \log_2\det\left( \mathbf{I} +\boldsymbol{\Sigma}_{\mathbf{e}_{\text{CP}}}^{-1} \hat{\mathbf{G}}_{\text{CP}}\boldsymbol{\Sigma}_{\mathbf{x}}\hat{\mathbf{G}}_{\text{CP}}^H \right). \label{eq:sum-rate-lower-bound}
\end{align}

Our goal is to optimize the INP strategies characterized by the linear combining matrices $\{\mathbf{U}, \mathbf{V}\}$ and the fronthaul quantization noise covariance matrices $\boldsymbol{\Omega}$ to maximize $f_{\text{sum}}(\mathbf{U}, \mathbf{V}, \boldsymbol{\Omega})$ under the fronthaul capacity constraints (\ref{eq:fronthaul-constraint-general}) across all APs. The problem can be formulated as
\begin{subequations} \label{eq:problem-original}
\begin{align}
    \underset{\mathbf{U}, \mathbf{V}, \boldsymbol{\Omega}} {\mathrm{max.}}\,\,\, & f_{\text{sum}}\left(\mathbf{U}, \mathbf{V}, \boldsymbol{\Omega}\right) \, \label{eq:problem-original-cost} \\
 \mathrm{s.t. }\,\,\,\,\,\, & f_{\text{FH},m,i}\left(\mathbf{U}, \mathbf{V}, \boldsymbol{\Omega}\right) \leq C_F, \, m\in\mathcal{M}, i\in\mathcal{L}, \label{eq:problem-original-fronthaul-capacity} \\
 & \boldsymbol{\Omega}_{m,i} \succeq \mathbf{0}, \, m\in\mathcal{M},i\in\mathcal{L}. \label{eq:problem-original-psd}
\end{align}
\end{subequations}
However, jointly optimizing $\{\mathbf{U},\mathbf{V}, \boldsymbol{\Omega}\}$ according to problem (\ref{eq:problem-original}) poses two main challenges. Firstly, the substantial number of optimization variables $ML K (N+2K)$ 
and constraints $ML$ involved leads to high computational complexity in finding a solution. Secondly, addressing (\ref{eq:problem-original}) with widely used iterative algorithms such as weighted MMSE \cite{Christensen:TWC08, Shi:TSP11, Zhao:TSP23-1, Zhao:TSP23-2, Yoo:WCL24} or fractional programming \cite{Shen:TSP18-1, Shen:TSP18-2, Chen:TSP24, Shen:JSAC24} demands the availability of the global CSI estimates $\hat{\mathbf{H}}_{m,i}$, $\forall m\in\mathcal{M},i\in\mathcal{L}$, at the CP for centralization optimization. For this reason, albeit the sequential operations of the proposed INP, its optimization formalism in \eqref{eq:problem-original} generally needs centralized calculations such as the CSI collection and optimization.
To tackle this issue, we propose an efficient successive design approach in the following.

\section{Proposed Sequential INP Optimization} \label{sec:sequential-optimization}

In this section, we present a sequential optimization algorithm for the proposed INP approach which calculates optimized linear combiners and quantization noise covariance matrices sequentially from AP $(m,1)$ to AP $(m,L)$, $\forall m\in\mathcal{M}$. We first discuss the design strategy of the linear combiners $\mathbf{U}_{m,i}$ and $\mathbf{V}_{m,i}$, which is followed by the optimization procedure of the fronthaul quantization $\boldsymbol{\Omega}_{m,i}$.
Specifically, for each stripe $m\in\mathcal{M}$, we follow a sequential optimization order: $\mathbf{U}_{m,1}\rightarrow\boldsymbol{\Omega}_{m,1}\rightarrow\{\mathbf{U}_{m,2},\mathbf{V}_{m,2}\}\rightarrow\boldsymbol{\Omega}_{m,2}\rightarrow\{\mathbf{U}_{m,3},\mathbf{V}_{m,3}\}\rightarrow\boldsymbol{\Omega}_{m,3}\rightarrow\ldots\rightarrow\{\mathbf{U}_{m,L},\mathbf{V}_{m,L}\}\rightarrow\boldsymbol{\Omega}_{m,L}$.

\subsection{Design of Linear combining matrices $\{\mathbf{U}, \mathbf{V}\}$} \label{sub:design-linear-combiners}
For sequential optimization, AP $(m,i)$ identifies its linear combiners $\{\mathbf{U}_{m,i},\mathbf{V}_{m,i}\}$ for fixed INP strategies $\{\mathbf{U}_{m,j},\mathbf{V}_{m,j}, \boldsymbol{\Omega}_{m,j}\}_{j=1}^{i-1}$ computed by preceding APs $(m,j)$, $\forall j=1,\cdots,i-1$.

We first define the vector $\tilde{\mathbf{y}}_{m,i} = [\mathbf{y}_{m,i}^H \, \mathbf{r}_{m,i-1}^H]^H\in\mathbb{C}^{(N+K)\times 1}$
by stacking the linear combining input signals at AP $(m,i)$.
The vector $\tilde{\mathbf{y}}_{m,i}$ can then be rewritten by
\begin{align}
    \tilde{\mathbf{y}}_{m,i} = \mathbf{B}_{m,i} \mathbf{x} + \tilde{\mathbf{w}}_{m,i}, \label{eq:INP-input-stacked}
\end{align}
where $\mathbf{B}_{m,i}\in\mathbb{C}^{(N+K)\times K}$ is defined as
\begin{align}
\mathbf{B}_{m,i} = [\hat{\mathbf{H}}_{m,i}^H , \hat{\mathbf{G}}_{m,i-1}^H]^H, \label{eq:Bmi}
\end{align}
and the covariance matrix $\boldsymbol{\Sigma}_{\tilde{\mathbf{w}}_{m,i}}$ of $\tilde{\mathbf{w}}_{m,i} = [\mathbf{w}_{m,i}^H , \mathbf{e}_{m,i-1}^H]^H\in\mathbb{C}^{(N+K)\times 1}$ is equal to
\begin{align}
    \boldsymbol{\Sigma}_{\tilde{\mathbf{w}}_{m,i}} = \mathbb{E}[\tilde{\mathbf{w}}_{m,i}\tilde{\mathbf{w}}_{m,i}^{H}]=\text{blkdiag}(\boldsymbol{\Sigma}_{\mathbf{w}_{m,i}}, \boldsymbol{\Sigma}_{\mathbf{e}_{m,i-1}}).\label{eq:w_tilde_cov}
\end{align}
The linear combining operation at AP $(m,i)$ in (\ref{eq:r_tilde}) can be interpreted as a linear estimation of $\mathbf{x}$ from $\tilde{\mathbf{y}}_{m,i}$ in (\ref{eq:INP-input-stacked}) by applying a linear filter $\mathbf{A}_{m,i} = [\mathbf{U}_{m,i}^H , \mathbf{V}_{m,i}^H]^H \in \mathbb{C}^{(N+K)\times K}$, yielding
\begin{align}
    \tilde{\mathbf{r}}_{m,i} &= \mathbf{A}_{m,i}^H \tilde{\mathbf{y}}_{m,i} \nonumber \\ 
    & = \mathbf{A}_{m,i}^H\left( \mathbf{B}_{m,i} \mathbf{x} + \tilde{\mathbf{w}}_{m,i} \right). \nonumber
\end{align}
Given the independence between the target signal $\mathbf{x}$ and the noise component $\tilde{\mathbf{w}}_{m,i}$, the linear MMSE estimator can maximize the sum-rate $I(\mathbf{x}; \tilde{\mathbf{r}}_{m,i})$ \cite{Kailath:00}. According to the orthogonality principle \cite{Kay:93}, the linear MMSE estimator $\mathbf{A}_{m,i}$ can be derived as
\begin{align}
    \mathbf{A}_{m,i} &= \left( \mathbb{E}\left[ \tilde{\mathbf{y}}_{m,i}\tilde{\mathbf{y}}_{m,i}^H \right] \right)^{-1} \mathbb{E}\left[ \tilde{\mathbf{y}}_{m,i} \mathbf{x}^H \right]\label{eq:optimal-UV-MMSE} \\
    &= \left(\mathbf{B}_{m,i}\boldsymbol{\Sigma}_{\mathbf{x}}\mathbf{B}_{m,i}^H + \boldsymbol{\Sigma}_{\tilde{\mathbf{w}}_{m,i}}\right)^{-1} \mathbf{B}_{m,i}\boldsymbol{\Sigma}_{\mathbf{x}}. \nonumber
\end{align}
Once $\mathbf{A}_{m,i}$ is computed using (\ref{eq:optimal-UV-MMSE}), the optimal combining matrices $\mathbf{U}_{m,i}$ and $\mathbf{V}_{m,i}$ can be retrieved by taking the first $N$ row vectors and the last $K$ row vectors of $\mathbf{A}_{m,i}$, respectively.

Notice that the local CSI estimate $\hat{\mathbf{H}}_{m,i}$ and the error covariance $\boldsymbol{\Sigma}_{\mathbf{w}_{m,i}}$ in \eqref{eq:covariance-Rx-signal} are available at AP $(m,i)$ through the uplink channel estimation procedure. Therefore, the additional knowledge requested for calculating $\mathbf{A}_{m,i}$ in \eqref{eq:optimal-UV-MMSE} includes the matrices $\hat{\mathbf{G}}_{m,i-1}$ and $\boldsymbol{\Sigma}_{\mathbf{e}_{m,i-1}}$, which can be shared from preceding AP $(m,i-1)$. This sequential AP interaction protocol will be discussed in Sec. \ref{sub:sequential-implementation}.

\subsection{Design of Quantization Noise Covariance Matrix $\boldsymbol{\Omega}$} \label{sub:design-quantization}

Next, we design the fronthaul quantization noise covariance $\boldsymbol{\Omega}_{m,i}$ of AP $(m,i)$, given its linear combiners $\{\mathbf{U}_{m,i}, \mathbf{V}_{m,i}\}$ in \eqref{eq:optimal-UV-MMSE} and the optimized INP strategies of preceding APs $\{\mathbf{U}_{m,j},\mathbf{V}_{m,j}, \boldsymbol{\Omega}_{m,j}\}_{j=1}^{i-1}$.
The combining output signal $\tilde{\mathbf{r}}_{m,i}$ in (\ref{eq:r_tilde}), related to its quantized version ${\mathbf{r}}_{m,i}$ in (\ref{eq:rewritting-r}) by $\tilde{\mathbf{r}}_{m,i} = {\mathbf{r}}_{m,i} - \mathbf{q}_{m,i}$, can be expressed as
\begin{align}
    \tilde{\mathbf{r}}_{m,i} & = \hat{\mathbf{G}}_{m,i} \mathbf{x} + \mathbf{e}_{m,i} - \mathbf{q}_{m,i}, \label{eq:linear-combining-output-rewritten} \\
    &= \hat{\mathbf{G}}_{m,i} \mathbf{x} + \mathbf{n}_{m,i}, \nonumber
\end{align}
where the covariance matrix $\boldsymbol{\Sigma}_{\mathbf{n}_{m,i}}$ of the noise component $\mathbf{n}_{m,i} = \mathbf{e}_{m,i} - \mathbf{q}_{m,i} =\mathbf{U}_{m,i}^H\mathbf{w}_{m,i} + \mathbf{V}_{m,i}^H\mathbf{e}_{m,i-1}$ is given by
\begin{align}
    &\boldsymbol{\Sigma}_{\mathbf{n}_{m,i}} = \mathbf{U}_{m,i}^H\boldsymbol{\Sigma}_{\mathbf{w}_{m,i}}\mathbf{U}_{m,i} + \mathbf{V}_{m,i}^H \boldsymbol{\Sigma}_{\mathbf{e}_{m,i-1}} \mathbf{V}_{m,i} \label{eq:covariance-n-mi} \\
    &= \sum\nolimits_{j=1}^i \mathbf{G}_{m,i,j}\boldsymbol{\Sigma}_{\mathbf{w}_{m,j}}\mathbf{G}_{m,i,j}^H + \sum\nolimits_{j=1}^{i-1} \mathbf{T}_{m,i,j}\boldsymbol{\Omega}_{m,j}\mathbf{T}_{m,i,j}^H. \nonumber
\end{align}
Note that unlike $\boldsymbol{\Sigma}_{\mathbf{e}_{m,i}}$ in (\ref{eq:covariance-e-i}), the summation index $j$ in the second term runs only up to $j=i-1$ rather than $j=i$.

To maximize the sum-rate in \eqref{eq:sum-rate-lower-bound}, the quantization noise covariance matrix
$\boldsymbol{\Omega}_{m,i}$ at AP $(m,i)$ can be determined for maximizing the mutual information $I(\mathbf{x}; \mathbf{r}_{m,i})$, which quantifies the achievable sum-rate when the transmitted data $\mathbf{x}$ is decoded based on the quantized signal $\mathbf{r}_{m,i}$. Combining \eqref{eq:rewritting-r} and \eqref{eq:linear-combining-output-rewritten}, $\mathbf{r}_{m,i}$ can be rewritten as
\begin{align}  
    \mathbf{r}_{m,i} = \hat{\mathbf{G}}_{m,i} \mathbf{x} + \mathbf{n}_{m,i} + \mathbf{q}_{m,i}. \nonumber
\end{align}
This results in the lower bound on the expected sum-rate $\mathbb{E}_{\tilde{\mathbf{h}}}[I(\mathbf{x};\mathbf{r}_{m,i})]$ expressed by
\begin{align}
    f_{m,i}(\boldsymbol{\Omega}_{m,i})=\log_2\det\!\left(\mathbf{I} + \!\left(\boldsymbol{\Sigma}_{\mathbf{n}_{m,i}} \! + \boldsymbol{\Omega}_{m,i}\right)^{\!-1}\hat{\mathbf{G}}_{m,i}\boldsymbol{\Sigma}_{\mathbf{x}}\hat{\mathbf{G}}_{m,i}^H\right). \label{eq:f_mi}
\end{align}
Thanks to the sequential INP processing, AP $(m,i)$ can know the optimized quantization noise covariances $\{\boldsymbol{\Omega}_{m,j}\}_{j=1}^{i-1}$ of the preceding APs $(m,1),\ldots,(m,i-1)$ by means of the fronthaul coordination. As a result, $\boldsymbol{\Sigma}_{\mathbf{n}_{m,i}}$ in \eqref{eq:f_mi} can be regarded as a constant when optimizing $\boldsymbol{\Omega}_{m,i}$ at AP $(m,i)$. Thus, the corresponding sum-rate maximization problem can be formulated as
\begin{align}
    \underset{\boldsymbol{\Omega}_{m,i}} {\mathrm{max.}}\,\,\, & f_{m,i}(\boldsymbol{\Omega}_{m,i})  \, \label{eq:problem-successive-Omega-original} \\
 \mathrm{s.t. }\,\,\,\,\,\, & \log_2\det\left( \mathbf{I} + \boldsymbol{\Omega}_{m,i}^{-1}\left( \hat{\mathbf{G}}_{m,i}\boldsymbol{\Sigma}_{\mathbf{x}}\hat{\mathbf{G}}_{m,i}^H + \boldsymbol{\Sigma}_{\mathbf{n}_{m,i}} \right) \right) \leq C_F, \nonumber \\
 & \boldsymbol{\Omega}_{m,i} \succ \mathbf{0}. \nonumber
\end{align}
We assume that the noise covariance matrix $\boldsymbol{\Sigma}_{\mathbf{n}_{m,i}}$ is invertible, which holds as long as the preceding APs choose full-column rank combining matrices $\{\mathbf{U}_{m,j},\mathbf{V}_{m,j}, \boldsymbol{\Omega}_{m,j}\}_{j=1}^{i-1}$.
By applying the change of variable as $\tilde{\boldsymbol{\Omega}}_{m,i} = \mathbf{N}_{m,i}\boldsymbol{\Omega}_{m,i}\mathbf{N}_{m,i}^H$ with $\mathbf{N}_{m,i} = \boldsymbol{\Sigma}_{\mathbf{n}_{m,i}}^{-1/2}$, the problem (\ref{eq:problem-successive-Omega-original}) can be equivalently transformed into
\begin{align}
    \underset{\tilde{\boldsymbol{\Omega}}_{m,i}} {\mathrm{max.}}\,\,\, & \!\log_2\det\!\left(\!\mathbf{I} + \!\left(\mathbf{I} \! + \tilde{\boldsymbol{\Omega}}_{m,i}\right)^{\!-1}\!\mathbf{N}_{m,i}\hat{\mathbf{G}}_{m,i}\boldsymbol{\Sigma}_{\mathbf{x}}\hat{\mathbf{G}}_{m,i}^H\mathbf{N}_{m,i}^H\right)  \nonumber \\ 
 \mathrm{s.t. }\,\,\,\,\,\, & \log_2\det\left( \mathbf{I} + \tilde{\boldsymbol{\Omega}}_{m,i}^{-1}\left( \mathbf{N}_{m,i}\hat{\mathbf{G}}_{m,i}\boldsymbol{\Sigma}_{\mathbf{x}}\hat{\mathbf{G}}_{m,i}^H \mathbf{N}_{m,i}^H + \mathbf{I} \right) \right) \nonumber \\
 & \,\,\,\,\,\,\,\leq C_F, \label{eq:problem-successive-Omega-modified-1} \\
 & \tilde{\boldsymbol{\Omega}}_{m,i} \succ \mathbf{0}. \nonumber
\end{align}
Once the optimal $\tilde{\boldsymbol{\Omega}}_{m,i}$ is found from (\ref{eq:problem-successive-Omega-modified-1}), the quantization noise covariance matrix $\boldsymbol{\Omega}_{m,i}$ can recovered as
$\boldsymbol{\Omega}_{m,i} \leftarrow \mathbf{N}_{m,i}^{-1}\tilde{\boldsymbol{\Omega}}_{m,i}\mathbf{N}_{m,i}^{-H}$.

To further simplify the problem (\ref{eq:problem-successive-Omega-modified-1}), we consider the eigenvalue decomposition (EVD) \cite{Strang:16} of $\mathbf{N}_{m,i}\hat{\mathbf{G}}_{m,i}\boldsymbol{\Sigma}_{\mathbf{x}}\hat{\mathbf{G}}_{m,i}^H\mathbf{N}_{m,i}^H$ as
\begin{align}   
\mathbf{N}_{m,i}\hat{\mathbf{G}}_{m,i}\boldsymbol{\Sigma}_{\mathbf{x}}\hat{\mathbf{G}}_{m,i}^H\mathbf{N}_{m,i}^H = \mathbf{U}_{m,i}^{\text{eig}} \boldsymbol{\Gamma}_{m,i}^{\text{eig}} (\mathbf{U}_{m,i}^{\text{eig}})^H, \label{eq:EVD-desired-cov}
\end{align}
where $\mathbf{U}_{m,i}^{\text{eig}}$ is the unitary eigenvector matrix satisfying $\mathbf{U}_{m,i}^{\text{eig}} (\mathbf{U}_{m,i}^{\text{eig}})^H = (\mathbf{U}_{m,i}^{\text{eig}})^H \mathbf{U}_{m,i}^{\text{eig}} = \mathbf{I}$, and 
$\boldsymbol{\Gamma}_{m,i}^{\text{eig}} = \text{diag}(\gamma_{m,i,1}^{\text{eig}}, \ldots, \gamma_{m,i,K}^{\text{eig}})$ is the diagonal eigenvalue matrix with $\gamma_{m,i,1}^{\text{eig}}\geq \ldots\geq \gamma_{m,i,K}^{\text{eig}}> 0$.
Then, we can construct an efficient solution to $\tilde{\boldsymbol{\Omega}}_{m,i}$ by using the eigenvector matrix $\mathbf{U}_{m,i}^{\text{eig}}$ as
\begin{align}
    \tilde{\boldsymbol{\Omega}}_{m,i} =  \mathbf{U}_{m,i}^{\text{eig}} \text{diag}\left( \omega_{m,i,1}, \ldots, \omega_{m,i,K}  \right)(\mathbf{U}_{m,i}^{\text{eig}})^H, \label{eq:EVD-Omega-tilde}
\end{align}
where $\omega_{m,i,k} > 0$ stands for the eigenvalue of $\tilde{\boldsymbol{\Omega}}_{m,i}$, which becomes a new optimization variable. Although this assumption does not guarantee optimality, it allows us to diagonalize the matrices inside the determinant functions in (\ref{eq:problem-successive-Omega-modified-1}), thereby leading to closed-form solutions.

By substituting (\ref{eq:EVD-desired-cov}) and (\ref{eq:EVD-Omega-tilde}) into (\ref{eq:problem-successive-Omega-modified-1}), we can obtain the following problem:
\begin{align}
    \underset{\boldsymbol{\omega}_{m,i}} {\mathrm{max.}}\,\,\, & \sum\nolimits_{k\in\mathcal{K}} \log_2\left( 1 + \frac{\gamma_{m,i,k}^{\text{eig}}}{1 + \omega_{m,i,k}} \right)\label{eq:problem-successive-Omega-modified-scalar} \\
 \mathrm{s.t. }\,\,\,\,\,\, & \sum\nolimits_{k\in\mathcal{K}}\log_2\left( 
1 + \frac{1 + \gamma_{m,i,k}^{\text{eig}}}{\omega_{m,i,k}} \right)\leq C_F, \nonumber \\
 & \omega_{m,i,k} > 0,\, \forall k\in\mathcal{K}. \nonumber
\end{align}
where $\boldsymbol{\omega}_{m,i} = \{\omega_{m,i,k}\}_{k\in\mathcal{K}}$. It is straightforward to show that the problem (\ref{eq:problem-successive-Omega-modified-scalar}) can be equivalently reformulated in terms of the variables $a_{m,i,k} = 1/\omega_{m,i,k}$, $k\in\mathcal{K}$, as
\begin{subequations} \label{eq:problem-successive-Omega-modified}
\begin{align} 
    \underset{\mathbf{a}_{m,i}} {\mathrm{max.}}\,\,\, & \sum\nolimits_{k\in\mathcal{K}} \log_2\left( 1 + a_{m,i,k}\left(\gamma_{m,i,k}^{\text{eig}} + 1\right)\right) \\ 
    & - \sum\nolimits_{k\in\mathcal{K}} \log_2\left( 1 + a_{m,i,k}\right) \, \nonumber \\
 \mathrm{s.t. }\,\,\,\,\,\, & \sum\nolimits_{k\in\mathcal{K}} \log_2\left( 1 + a_{m,i,k}\left(\gamma_{m,i,k}^{\text{eig}} + 1\right) \right)  \leq C_F, \label{eq:problem-successive-Omega-modified-fronthaul-capacity} \\
 & a_{m,i,k} \geq 0, \, k\in\mathcal{K}, \label{eq:problem-successive-Omega-modified-nonnegativity}
\end{align}
\end{subequations}
with $\mathbf{a}_{m,i} = \{a_{m,i,k}\}_{k\in\mathcal{K}}$.

The following theorem presents a closed-form solution to the problem (\ref{eq:problem-successive-Omega-modified}).
\begin{theorem}
    The solution $\mathbf{a}_{m,i}$ to the problem (\ref{eq:problem-successive-Omega-modified}) is given as
    \begin{align}
    a_{m,i,k} = \! \left[ \frac{1}{\lambda_{m,i}}\!\left(\!1 \!-\! \frac{1}{\gamma_{m,i,l}^{\text{eig}} + 1}\right) \!- 1 \right]^{+}\!\!\!\!, \label{eq:local-optimum}
    \end{align}
    for $k\in\mathcal{K}$, with $[\cdot]^{+} = \max\{\cdot,0\}$. Here, the optimal Lagrange multiplier $\lambda_{m,i}$ is chosen to satisfy
\begin{align}
    \sum\nolimits_{k\in\mathcal{K}} \log_2\left( 1 + a_{m,i,k}\left(\gamma_{m,i,k}^{\text{eig}} + 1\right) \right) = C_F. \label{eq:condition-optimal-Lagrange-multiplier}
\end{align}
Such $\lambda_{m,i}$ can be found via the bisection method \cite{Boyd:Cambridge}, since the left-hand side of (\ref{eq:condition-optimal-Lagrange-multiplier}) monotonically decreases with $\lambda_{m,i}$.
\end{theorem}
\begin{proof}
    Please refer to Appendix A for the proof.
\end{proof}

Once $\mathbf{a}_{m,i}$ is given from (\ref{eq:local-optimum}), the quantization noise covariance matrix $\boldsymbol{\Omega}_{m,i}$ is obtained as
\begin{align}
    \boldsymbol{\Omega}_{m,i} \leftarrow \mathbf{N}_{m,i}^{-1}\mathbf{U}_{m,i}^{\text{eig}}\text{diag}(\boldsymbol{\omega}_{m,i}) (\mathbf{U}_{m,i}^{\text{eig}})^H \mathbf{N}_{m,i}^{-H}, \label{eq:recover-Omega-from-a}
\end{align}
with $\omega_{m,i,k} = 1 / a_{m,i,k}$, $k\in\mathcal{K}$.

\begin{figure*}[t]
\centering\includegraphics[width=0.7\linewidth]{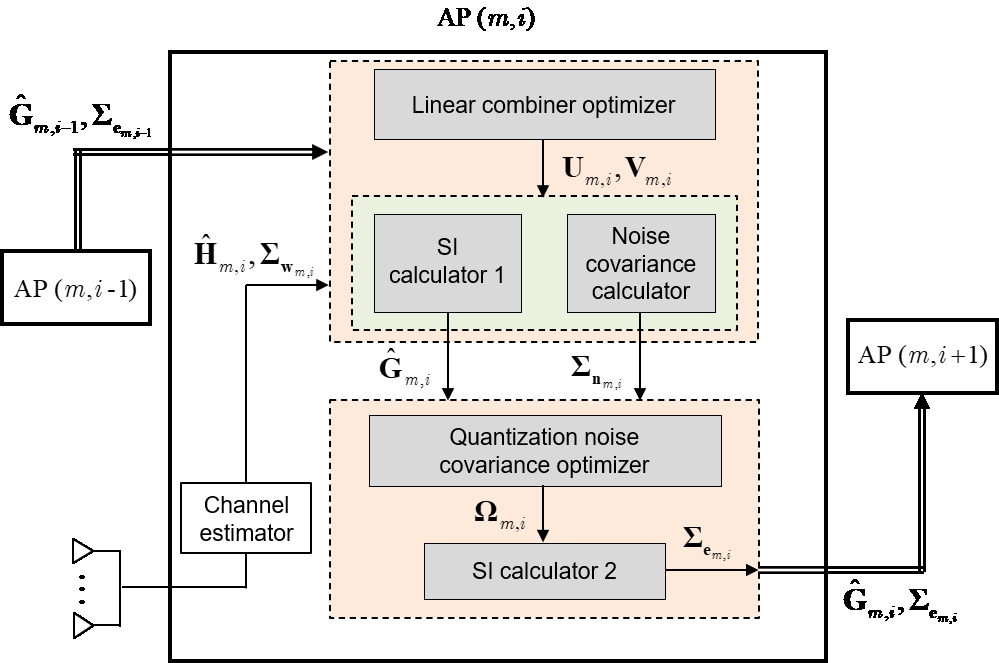}
\caption{\label{fig:opt-SI-perAP}Illustration of procedures for AP $(m,i)$ to optimizing its INP strategy and share the SI with AP $(m,i-1)$ and AP $(m,i+1)$.}
\end{figure*}

\subsection{Sequential Implementation} \label{sub:sequential-implementation}

We provide a sequential implementation of the proposed INP optimization process in Algorithm 1.
The proposed algorithm can be executed in parallel for each of radio stripes $m\in\mathcal{M}$, while APs connected through the identical radio stripe 
sequentially compute (\ref{eq:optimal-UV-MMSE}) and (\ref{eq:local-optimum})
%
from AP $(m,1)$ to AP $(m,L)$, $\forall m\in\mathcal{M}$. AP $(m,1)$, $\forall m\in\mathcal{M}$, begins the sequential optimization procedure with initial $\hat{\mathbf{G}}_{m,0}$ and $\boldsymbol{\Sigma}_{\mathbf{e}_{m,0}}$ being empty matrices. Each AP $(m,i)$ first obtains the combining matrices $\{\mathbf{U}_{m,i}, \mathbf{V}_{m,i}\}$ from \eqref{eq:optimal-UV-MMSE}. To this end, AP $(m,i)$ calculates $\mathbf{B}_{m,i}$ and $\boldsymbol{\Sigma}_{\tilde{\mathbf{w}}_{m,i}}$ from \eqref{eq:Bmi} and \eqref{eq:w_tilde_cov}, respectively. This requires the local CSI $\{\hat{\mathbf{H}}_{m,i}, \boldsymbol{\Sigma}_{\mathbf{w}_{m,i}}\}$ along with the $\{\hat{\mathbf{G}}_{m,i-1}, \boldsymbol{\Sigma}_{\mathbf{e}_{m,i-1}}\}$ which can be shared from AP $(m,i-1)$. Such a sequential AP coordination through radio stripes facilitates efficient fronthaul utilization compared to naive case which exchanges local CSI matrices $\{\hat{\mathbf{H}}_{m,j}, \boldsymbol{\Sigma}_{\mathbf{w}_{m,j}}\}_{j=1}^{i}$ directly.

Once the optimized linear combiners are obtained, AP $(m,i)$ computes $\{\hat{\mathbf{G}}_{m,i}, \boldsymbol{\Sigma}_{\mathbf{n}_{m,i}}\}$ with (\ref{eq:definition-Ghat}) and (\ref{eq:covariance-n-mi}), thereby optimizing the quantization noise covariance $\mathbf{\Omega}_{m,i}$. 
Lastly, AP $(m,i)$ evaluates the effective noise covariance matrix $\boldsymbol{\Sigma}_{\mathbf{e}_{m,i}} = \boldsymbol{\Sigma}_{\mathbf{n}_{m,i}} + \boldsymbol{\Omega}_{m,i}$ and transmits the SI $\{\hat{\mathbf{G}}_{m,i}, \boldsymbol{\Sigma}_{\mathbf{e}_{m,i}}\}$ to the proceeding AP $(m,i+1)$ for the optimization of $\{\mathbf{U}_{m,i+1}, \mathbf{V}_{m,i+1}, \mathbf{\Omega}_{m,i+1}\}$.
This iterative process continues until $L$ sequential optimization steps are completed. The detailed procedures for AP $(m,i)$ to optimize its INP strategy and share the SI with neighboring APs are depicted in Fig.~\ref{fig:opt-SI-perAP}.
\\

\begin{algorithm}
\caption{Proposed sequential INP design algorithm}

\textbf{Input:} $M$, $L$, $K$, $N$, $\boldsymbol{\Sigma}_{\mathbf{x}}$, $\{ \hat{\mathbf{H}}_{m,i}, \boldsymbol{\Sigma}_{\mathbf{w}_{m,i}} \}_{m\in\mathcal{M}, i\in\mathcal{L}}$, $C_F$

\textbf{\footnotesize{}1}~\textbf{ for each radio stripe} $m\in\mathcal{M}$ \textbf{ in parallel}

\textbf{\footnotesize{}2}~~~~\textbf{ for each AP} $i\in\mathcal{L}$ \textbf{ (sequentially)}

\textbf{\footnotesize{}3}~~~~~~~~\textbf{Linear Combiners:}

\textbf{\footnotesize{}4}~~~~~~~~AP $(m,i)$ sets $\mathbf{B}_{m,i} \leftarrow [\hat{\mathbf{H}}_{m,i}^H \, \hat{\mathbf{G}}_{m,i-1}^H]^H$.

\textbf{\footnotesize{}5}~~~~~~~~AP $(m,i)$ sets $\boldsymbol{\Sigma}_{\tilde{\mathbf{w}}_{m,i}} \leftarrow \text{blkdiag}(\boldsymbol{\Sigma}_{\mathbf{w}_{m,i}}, \boldsymbol{\Sigma}_{\mathbf{e}_{m,i-1}})$.

\textbf{\footnotesize{}6}~~~~~~~~AP $(m,i)$ computes $\mathbf{A}_{m,i} = [\mathbf{U}_{m,i}^H \, \mathbf{V}_{m,i}^H]^H$ with (\ref{eq:optimal-UV-MMSE}).

\textbf{\footnotesize{}7}~~~~~~~~AP $(m,i)$ computes $\{\hat{\mathbf{G}}_{m,i}, \boldsymbol{\Sigma}_{\mathbf{n}_{m,i}}\}$ with (\ref{eq:definition-Ghat}) and (\ref{eq:covariance-n-mi}).

\textbf{\footnotesize{}8}~~~~~~~~\textbf{Quantization Noise Covariance Matrix:}

\textbf{\footnotesize{}9}~~~~~~~~AP $(m,i)$ computes $\{\gamma_{m,i,k}^{\text{eig}}\}_{k\in\mathcal{K}}$ with (\ref{eq:EVD-desired-cov}).

\textbf{\footnotesize{}10}~~~~~~AP $(m,i)$ computes $\mathbf{a}_{m,i}$ with (\ref{eq:local-optimum}).

\textbf{\footnotesize{}11}~~~~~~AP $(m,i)$ obtains $\boldsymbol{\Omega}_{m,i}$ with (\ref{eq:recover-Omega-from-a}).


\textbf{\footnotesize{}12}~~~~~~AP $(m,i)$ computes $\boldsymbol{\Sigma}_{\mathbf{e}_{m,i}} = \boldsymbol{\Sigma}_{\mathbf{n}_{m,i}} + \boldsymbol{\Omega}_{m,i}$.

\textbf{\footnotesize{}13}~~~~~~AP $(m,i)$ sends $\{\hat{\mathbf{G}}_{m,i}, \boldsymbol{\Sigma}_{\mathbf{e}_{m,i}}\}$ to AP $(m,i+1)$.

\textbf{Output:} $\{ \mathbf{U}_{m,i}, \mathbf{V}_{m,i}, \boldsymbol{\Omega}_{m,i} \}_{m\in\mathcal{M}, i\in\mathcal{L}}$

\end{algorithm}

\subsection{Hybrid Architecture With Limited RF Chains} \label{sub:hybrid-architecture}

So far, we have considered fully-digital signal processing methods assuming that the number of RF chains at APs is the same as that of antennas. However, in a practical mMIMO setup, equipping the RF chains for all AP antennas would be costly.
To address this, a hybrid analog and digital processing can be used \cite{Lee:TC17,Kim:TC19, Kassam:TC23}. It is assumed that each AP $(m,i)$ is equipped with only $K\ll N$ RF chains. The received signal vector $\mathbf{y}_{m,i}$ at AP $(m,i)$ is processed by a two-stage receive combining. Therefore, the hybrid combining matrix $\mathbf{U}_{m,i}^{\text{hyb}}$ in (\ref{eq:r_tilde}) is designed as
\begin{align}
    \mathbf{U}_{m,i}^{\text{hyb}} = \mathbf{U}_{m,i}^{A} \mathbf{U}_{m,i}^{D}, \label{eq:U_H}
\end{align}
where $\mathbf{U}_{m,i}^{A} \in \mathbb{C}^{N \times K}$ and $\mathbf{U}_{m,i}^{D} \in \mathbb{C}^{K \times K}$ represent the analog and digital combining matrices, respectively.
The analog combining generates a $K$-dimensional signal vector $(\mathbf{U}_{m,i}^{A})^H \mathbf{y}_{m,i}$ by combining
phase-shifted versions of the incoming signal vector $\mathbf{y}_{m,i}$. Therefore, each element of $\mathbf{U}_{m,i}^{A}$ must satisfy a constant-modulus constraint, i.e., $|\mathbf{U}_{m,i}^{A}(n,k)| = 1$, where $\mathbf{X}(n,k)$ denotes the $(n,k)$th element of a matrix $\mathbf{X}$.

We now discuss an extension of our proposed linear combining matrix design in Sec. \ref{sub:design-linear-combiners} to the hybrid analog-digital architecture.
To maintain scalability of the algorithm, we sequentially design the analog and digital combining matrices instead of jointly optimizing them.

\textbf{1) Design of Analog Combining Matrix $\mathbf{U}_{m,i}^{A}$:} The analog combining matrix $\mathbf{U}_{m,i}^{A}$ can be designed by using the fully-digital combining matrix $\mathbf{U}_{m,i}$ obtained from (\ref{eq:optimal-UV-MMSE}). Since the performance of hybrid analog-digital scheme is inherently upper bounded by the fully-digital scheme, we aim to achieve performance as close to the fully-digital solution as possible. To this end, we derive the analog combining matrix $\mathbf{U}_{m,i}^A$ by projecting the fully-digital solution onto the feasible space \cite{Lee:TC17,Kim:TC19}.
Since the constant-modulus constraints are imposed element-wise, the projected analog combining matrix is simply given by \cite{Kim:TC19}
\begin{align}
    \mathbf{U}_{m,i}^A(n,k) = \frac{\mathbf{U}_{m,i}(n,k)}{|\mathbf{U}_{m,i}(n,k)|}, \label{eq:projection-analog-combining}
\end{align}
for all $n\in\mathcal{N} = \{1,2,\ldots,N\}$ and $k\in\mathcal{K}$.

\textbf{2) Design of Digital Combining Matrix $\mathbf{U}_{m,i}^{D}$:} Since the projection process in (\ref{eq:projection-analog-combining}) may cause a deviation of $\mathbf{U}_{m,i}^A$ from the fully-digital solution $\mathbf{U}_{m,i}^{\text{full}}$, we design the digital combining matrix $\mathbf{U}_{m,i}^D$ to compensate for this deviation with the criterion:
\begin{align} 
    \underset{\mathbf{U}_{m,i}^D} {\mathrm{min.}}\,\,\, & \big\| \mathbf{U}_{m,i} - \mathbf{U}_{m,i}^A\mathbf{U}_{m,i}^D \big\|^2_F, \nonumber 
\end{align}
where $\mathbf{U}_{m,i}$ and $\mathbf{U}_{m,i}^A$ are fixed matrices. The optimal solution to this unconstrained quadratic convex problem is given by
\begin{align}
    \mathbf{U}_{m,i}^D = \left( \left(\mathbf{U}_{m,i}^A\right)^H \mathbf{U}_{m,i}^A \right)^{-1} \left(\mathbf{U}_{m,i}^A\right)^H \mathbf{U}_{m,i}. \label{eq:optimal-digital-in-hybrid}
\end{align}
Finally, the hybrid combining matrix $\mathbf{U}_{m,i}^{\text{hyb}}$ can be computed from \eqref{eq:U_H}. Notice that \eqref{eq:projection-analog-combining} and \eqref{eq:optimal-digital-in-hybrid} do not require the information of subsequent APs $(m,i+1),\cdots,(m,L)$. Consequently, the hybrid combining architecture can be readily applied to the proposed INP procedure in Algorithm 1.

\section{Numerical Results} \label{sec:numerical-results}

\subsection{Simulation setup} \label{sub:setup}

We validate the effectiveness of the proposed sequential INP design through numerical results. All simulations are realized with MATLAB. For the simulations, UEs are randomly located within a circular area of radius 200 m. 
We consider the cases with $M\in\{1,2,3,4\}$ radio stripes.
For the case of $M=1$, $L$ APs are evenly spaced around a circle surrounding the coverage area.
For $M\geq 2$, as shown in Fig. \ref{fig:system-model}, APs within each stripe are positioned at uniform intervals along their radius and arc.
The local scattering model in \cite[Sec. 2.6] {Bjornson:FnT17} is adopted for the spatial correlation matrix $\mathbf{R}_{m,i,k}$ with the path-loss $\beta_{m,i,k} = \text{tr}(\mathbf{R}_{m,i,k}))/N$ given as $\beta_{m,i,k}[\text{dB}] = -30.5 - 36.7\log_{10}d_{m,i,k}$, where $d_{m,i,k}$ is the distance between UE $k$ and AP $(m,i)$, accounting for a height difference of 5 m \cite{Shaik-et-al:TC21}.
Each AP is equipped with a uniform linear array with half-wavelength antenna spacing, denoted by $d_{\text{H}}$. The multipath components follow a Gaussian distribution in the angular domain with a standard deviation $\sigma_{\phi} = 15^{\circ}$ centered around a nominal angle, which is denoted by $\phi_{k}$ for each UE $k$.
Thus, the $(a,b)$th element of $\mathbf{R}_{m,i,k}$ are given by 
\begin{align}
    \mathbf{R}_{m,i,k}(a,b)=
    \beta_{m,i,k}\int e^{2\pi jd_{\text{H}}(a-b)\sin(\phi_k + \delta)}\frac{1}{\sqrt{2\pi}\sigma_{\phi}}e^{-\frac{\delta^{2}}{2\sigma^{2}_{\phi}}} d\delta. \nonumber
\end{align}
Unless stated otherwise, the transmit power of each UE is set to $P_k=50$ mW, and the noise power $\sigma^2_{z}$ is -85 dBm with a noise figure of 9 dB, operating over a bandwidth of 100 MHz.

\begin{table}
\caption{\small Description of baseline schemes and cutset upperbound.}
\vspace{-2mm}
\centering
\renewcommand{\arraystretch}{1.5}
\begin{tabular}{|c||c|c|}
    \hline    
     Baseline & INP component & Description \\
      \hline \hline
     \scriptsize{MRC \cite{Ngo:TWC17, Zhang:Acc18}}  & \scriptsize{Linear comb.} & $\mathbf{U}_{m,i} \leftarrow \hat{\mathbf{H}}_{m,i}$, $\mathbf{V}_{m,i} \leftarrow \mathbf{I}$ \\
      \hline
     \multirow{2}{*}{\scriptsize{naiveFH \cite{Liu:Cambridge17, Zhang:arXiv24}}} & \multirow{2}{*}{\scriptsize{Fronth. comp.}} & $\boldsymbol{\Omega}_{m,i} \leftarrow \text{diag}(\{d_{m,i,k}\}_{k\in\mathcal{K}})$ \\     
     & & with $d_{m,i,k}$ in (\ref{eq:naiveFH})  \\
    \hline    
     \scriptsize{Cutset bound \cite{Lim:TIT11}} & \scriptsize{Capacity bound} & $\!\!\!\!\!\begin{array}{c}
          C_{\text{cutset}} = \!\!\!\!
 \min\limits_{\tilde{\mathcal{M}}\subseteq \mathcal{M}} \!\! \bigg\{ \! C_F(M - |\tilde{\mathcal{M}}|) \!\!\!\\
           \,\,\,\,\,+ I\left( \mathbf{x} ; \{\mathbf{y}_{m,i}\}_{m\in\tilde{\mathcal{M}}, i\in\mathcal{L}} \right) 
 \bigg\} \!\!\!
     \end{array}$
     \\
     \hline
\end{tabular}
\end{table}

To see the individual impacts of the proposed MMSE linear combining in Sec. \ref{sub:design-linear-combiners} and the optimized fronthaul (optFH) compression in Sec. \ref{sub:design-quantization}, we consider the following benchmark schemes.
\begin{itemize}
    \item Maximum-Ratio Combining (MRC) \cite{Ngo:TWC17, Zhang:Acc18}: The linear combining matrices are set to $\mathbf{U}_{m,i}= \hat{\mathbf{H}}_{m,i}$ and $\mathbf{V}_{m,i}= \mathbf{I}$.
    \item Naive FrontHaul compression (naiveFH) \cite{Liu:Cambridge17, Zhang:arXiv24}: The elements of $\tilde{\mathbf{r}}_{m,i}$ are individually quantized with an equal fronthaul rate allocation. Consequently, $\boldsymbol{\Omega}_{m,i}$ is set to $\boldsymbol{\Omega}_{m,i} = \text{diag}(\{d_{m,i,k}\}_{k\in\mathcal{K}})$, where each diagonal element $d_{m,i,k}$ is given by
    \begin{align}
        d_{m,i,k} = \frac{1}{2^{C_F/K}-1} \left[ \hat{\mathbf{G}}_{m,i} \boldsymbol{\Sigma}_{\mathbf{x}} \hat{\mathbf{G}}_{m,i}^H + \boldsymbol{\Sigma}_{\mathbf{n}_{m,i}} \right]_{k,k}, \label{eq:naiveFH}
    \end{align}
    with $[\cdot]_{k,k}$ taking the $k$th diagonal element.
\end{itemize}
For reference, we also evaluate the performance gap to a cutset upperbound \cite{Lim:TIT11}.
The cutset bound for the CF-mMIMO system described in Sec. \ref{sec:System-Model} is given as
\begin{align}
    C_{\text{cutset}} = \min_{\tilde{\mathcal{M}}\subseteq \mathcal{M}} \bigg\{  C_F(M - |\tilde{\mathcal{M}}|) + I\left( \mathbf{x} ; \{\mathbf{y}_{m,i}\}_{m\in\tilde{\mathcal{M}}, i\in\mathcal{L}} \right)
 \bigg\} \nonumber
\end{align}
with the mutual information 
$I ( \mathbf{x} ; \{\mathbf{y}_{m,i}\}_{m\in\tilde{\mathcal{M}}, i\in\mathcal{L}} )$ being
\begin{align}
   I\left( \mathbf{x} ; \{\mathbf{y}_{m,i}\}_{m\in\tilde{\mathcal{M}}, i\in\mathcal{L}} \right) = \log_2\det\left( \mathbf{I} + \bar{\boldsymbol{\Sigma}}_{\mathbf{w}_{\tilde{\mathcal{M}}}}^{-1} \bar{\mathbf{H}}_{\tilde{\mathcal{M}}}\boldsymbol{\Sigma}_{\mathbf{x}}\bar{\mathbf{H}}_{\tilde{\mathcal{M}}}^H \right), \nonumber
\end{align}
where $\bar{\mathbf{H}}_{\tilde{\mathcal{M}}} \in \mathbb{C}^{L N \tilde{M} \times K}$ stacks $\bar{\mathbf{H}}_m$, $m\in\tilde{\mathcal{M}}$, vertically with $\bar{\mathbf{H}}_m = [\hat{\mathbf{H}}_{m,1}^H \, \cdots \, \hat{\mathbf{H}}_{m,L}^H]^H \in\mathbb{C}^{LN \times K}$, and $\bar{\boldsymbol{\Sigma}}_{\mathbf{w}_{\tilde{\mathcal{M}}}} = \text{blkdiag}( \{\bar{\boldsymbol{\Sigma}}_{\mathbf{w}_{m}}\}_{m\in\tilde{\mathcal{M}}} )$ with $\bar{\boldsymbol{\Sigma}}_i = \text{blkdiag}(\{\boldsymbol{\Sigma}_{m,i}\}_{i\in\mathcal{L}})$.
The above baseline schemes and the cutset upperbound are summarized in Table 2.

Based on the above benchmark schemes, we can consider three different combinations among the linear combining methods, i.e., MMSE and MRC, and fronthaul compression strategies, i.e., optFH and naiveFH. The proposed INP algorithm can be interpreted as ``MMSE + optFH''. By comparing the proposed approach with ``MRC \cite{Ngo:TWC17, Zhang:Acc18} + optFH'' and ``MMSE + naiveFH \cite{Liu:Cambridge17, Zhang:arXiv24}'' methods, we can examine the impacts of the proposed linear combiners and fronthaul quantization strategies, respectively. Notice that ``MMSE + naiveFH \cite{Liu:Cambridge17, Zhang:arXiv24}'' method can be viewed as a straightforward extension of the conventional scheme in \cite{Shaik-et-al:TC21}, which was originally developed for the single radio stripe case with infinite-capacity fronthaul links, to the considered parallel radio stripes systems with finite-capacity fronthual links.

\subsection{Complexity and Signaling Overhead} \label{sub:complexity-overhead}

In this subsection, we compare the complexity and signaling overhead of the proposed scheme with those of the baseline 
schemes described in Sec. \ref{sub:setup}.

\textbf{1) Complexity comparison:} We first compare the computational complexity of the proposed linear combining scheme in Sec. IV-A with that of the MRC scheme \cite{Ngo:TWC17, Zhang:Acc18}. Since all APs execute the same computations, we focus on the complexity at each AP $(m,i)$.
In the MRC scheme, the linear combining matrices $\mathbf{U}_{m,i}$ and $\mathbf{V}_{m,i}$ are simply set to the estimated local channel matrix and an identity matrix, respectively, resulting in a constant complexity $\mathcal{O}(1)$ which does not scale with any of system parameters.
In contrast, the proposed linear combining scheme in Sec. IV-A computes the linear combining matrices $\mathbf{U}_{m,i}$ and $\mathbf{V}_{m,i}$ as described in (\ref{eq:optimal-UV-MMSE}).
The dominant complexity here arises from the matrix inversion $( \mathbf{B}_{m,i} \boldsymbol{\Sigma}_{\mathbf{x}} \mathbf{B}_{m,i}^H + \boldsymbol{\Sigma}_{\tilde{\mathbf{w}}_{m,i}} )^{-1}$, which has a complexity of $\mathcal{O}( (N+K)^3 )$ using the Gauss-Jordan elimination method.

Next, we compare the complexity of the proposed fronthaul quantization strategy in Sec. IV-B with that of the naiveFH scheme \cite{Liu:Cambridge17, Zhang:arXiv24}.
The complexity of computing the quantization noise covariance matrix $\boldsymbol{\Omega}_{m,i}$ in the proposed scheme is dominated by the EVD in (\ref{eq:EVD-desired-cov}), which has a complexity of $\mathcal{O}(K^3)$ using the divide-and-conquer algorithm. Similarly, the complexity of the naiveFH scheme is dominated by the matrix multiplication $\hat{\mathbf{G}}_{m,i} \boldsymbol{\Sigma}_{\mathbf{x}} \hat{\mathbf{G}}_{m,i}^H$, which also has a complexity of $\mathcal{O}(K^3)$.

In summary, the proposed INP scheme matches the naiveFH scheme \cite{Liu:Cambridge17, Zhang:arXiv24} in terms of asymptotic complexity for computing the fronthaul quantization noise covariance matrix, while its complexity for deriving the linear combining matrices is higher than that of the MRC scheme \cite{Ngo:TWC17, Zhang:Acc18}, enabling significantly improved performance.

\textbf{2) Signaling overhead comparison:} In the implementation of the proposed INP scheme, as depicted in Fig. \ref{fig:opt-SI-perAP}, the SI matrices $\hat{\mathbf{G}}_{m,i} \in \mathbb{C}^{K\times K}$ and $\boldsymbol{\Sigma}_{\mathbf{e}_{m,i}} \in \mathbb{C}^{K\times K}$ need to be transmitted from each AP $(m,i)$ to the subsequent AP $(m,i+1)$.
These SI matrices, along with the local channel matrix $\hat{\mathbf{H}}_{m,i}$, are sufficient for computing both the linear combining matrices and the fronthaul quantization noise covariance matrix.
Consequently, the signaling overhead for each fronthaul link between neighboring APs amounts to $2K^2$ real values.

For the MRC scheme \cite{Ngo:TWC17, Zhang:Acc18}, only the estimated local channel matrix is needed to calculate the linear combining matrices, eliminating the need for additional signaling overhead. However, in the naiveFH scheme \cite{Liu:Cambridge17, Zhang:arXiv24}, the SI matrices 
$\hat{\mathbf{G}}_{m,i}$ and $\boldsymbol{\Sigma}_{\mathbf{e}_{m,i}}$ need to be exchanged over each fronthaul link, similar to the proposed INP scheme, resulting in a signaling overhead of $2K^2$ real values per link.

Table 3 summarizes the discussed computational complexity and signaling overhead of the proposed INP scheme, as well as the baseline MRC and naiveFH schemes.

\begin{table}
\caption{\small A comparison with the baseline schemes in terms of the complexity and signaling overhead.}
\vspace{-2mm}
\centering
\renewcommand{\arraystretch}{1.5}
\begin{tabular}{|c|c||c|c|}
    \hline    
     \multicolumn{2}{|c||}{ } &  Complexity  & Signal. overhead \\
     \hline \hline
     Linear & MRC  &  $\mathcal{O}(1)$ & None \\
    \cline{2-4}
      combining &  Proposed & $\mathcal{O}((N+K)^3)$ & $2K^2$ real values\\
    \hline \hline
    Fronthaul  & naiveFH  & $\mathcal{O}(K^3)$  & $2K^2$ real values \\
    \cline{2-4}
      quant. &  Proposed & $\mathcal{O}(K^3)$ &  $2K^2$ real values \\
    \hline    
\end{tabular}
\end{table}

\subsection{Simulation Results} \label{sub:simulation-results}

\begin{figure}
\centering\includegraphics[width=1.0\linewidth]{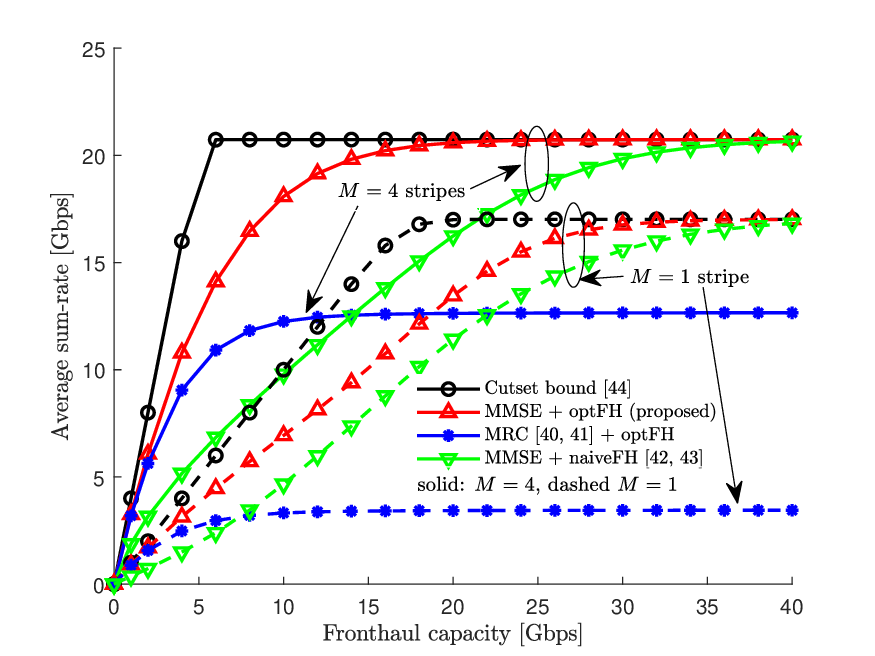}
\caption{\label{fig.avgRsum-vs-CF}Average sum-rate versus the fronthaul capacity $C_F$ with $M \in \{1, 4\}$, $L=32/M$, $N=24$, and $K=20$.}
\end{figure}

Fig. \ref{fig.avgRsum-vs-CF} plots the average sum-rate as a function of the fronthaul capacity $C_F$ for $M \in \{1, 4\}$ radio stripes, $L=32/M$ APs on each stripe, $N=24$ AP antennas, and $K=20$ UEs. 
We can see that the proposed INP approach, which jointly optimizes the linear combining and quantization noise covariance matrices, yields substantial improvements over baseline schemes. Increasing the number of radio stripes enhances overall performance at the expense of increased deployment costs for fronthaul cables. Also, as the fronthaul capacity $C_F$ increases, the sum-rates of all schemes grow due to the reduction in quantization noise levels.
Notably, in the small fronthaul capacity regime, ``MRC \cite{Ngo:TWC17, Zhang:Acc18} + optFH'' method performs better than ``MMSE + naiveFH \cite{Liu:Cambridge17, Zhang:arXiv24}'' method. This indicates that for a small $C_{F}$, the quantization strategy demonstrates a more significant impact than the linear combiners. In contrast, when we have sufficient fronthaul resources, the linear combining procedure plays a pivotal role in maximizing the sum-rate performance. By optimizing the quantization and combining strategies jointly, the proposed INP design outperforms all other benchmark schemes and shows the smallest gap to the cutset upperbound \cite{Lim:TIT11}.

\begin{figure}
\centering\includegraphics[width=1.0\linewidth]{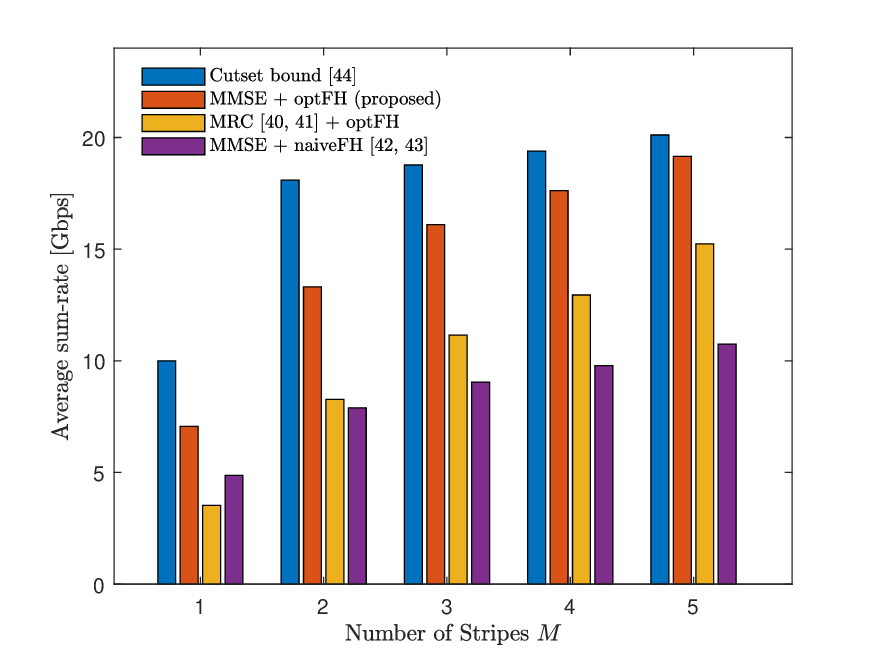}
\caption{\label{fig:avgRsum-vs-stripe}Average sum-rate versus the number of stripes $M$ with $L=24/M$, $N=24$, $C_F=10$, and $K=20$.}
\end{figure}

Fig. \ref{fig:avgRsum-vs-stripe} plots the average sum-rate versus the number of stripes $M$ for $L=24/M$ APs per stripe, $N=24$ AP antennas, $C_F=10$ Gbps, and $K=20$ UEs.
The performance of all schemes exhibits enhancement as the number of stripes increases. This is attributed to a more uniform distribution of APs across the coverage area, consequently leading to increased SNR levels at the receive antennas of APs.
It is noteworthy that the proposed sequential INP design approaches the cutset upperbound \cite{Lim:TIT11} as $M$ grows.
Comparing the proposed INP design and the baseline schemes, the impact of the fronthaul quantization becomes more pronounced with larger $M$, while the impact of optimizing linear combining slightly decreases with $M$.

\begin{figure}
\centering\includegraphics[width=1.0\linewidth]{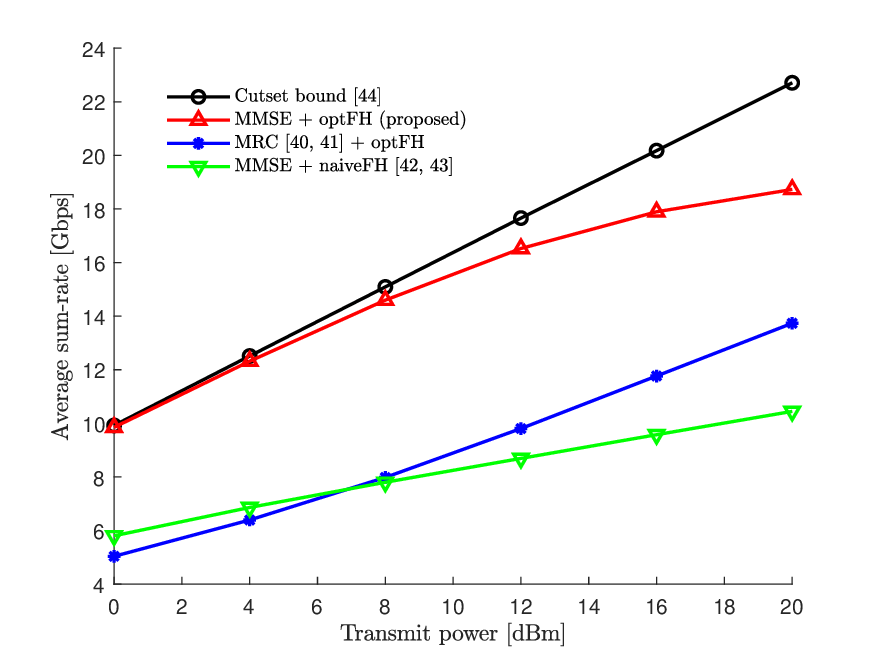}
\caption{\label{fig:avgRsum-vs-power}Average sum-rate versus transmit power $P_{\text{tx}}$ with $M=4$, $L=8$, $N=24$, $C_F=10$, and $K=20$.}
\end{figure}

Fig. \ref{fig:avgRsum-vs-power} shows the average sum rate with respect to the transmit power of UEs for $M=4$ stripes, $L=8$ APs per stripe each with $N=24$ antennas, $C_F=10$ Gbps and $K=20$ UEs. We set $P_k = P_{\text{tx}}$, $\forall k\in\mathcal{K}$, and thus the uplink signal-to-noise ratio (SNR) is defined as $P_{\text{tx}}/\sigma^{2}_{z}$. 
As the SNR increases, the impact of optimizing the fronthaul quantization noise covariance matrices becomes more pronounced, as the management of fronthaul quantization noise signals becomes increasingly crucial compared to additive noise signals.
Additionally, it is noteworthy that the sum-rate performance approaches the cutset upperbound \cite{Lim:TIT11} at low SNR levels only when both the linear combiners and the fronthaul quantization noise covariance are jointly optimized.
Specifically, at $P_{\text{tx}} = 8$ dBm, the proposed INP scheme achieves over $82.92$\% sum-rate gains compared to both baseline schemes.

\begin{figure}
\centering\includegraphics[width=1.0\linewidth]{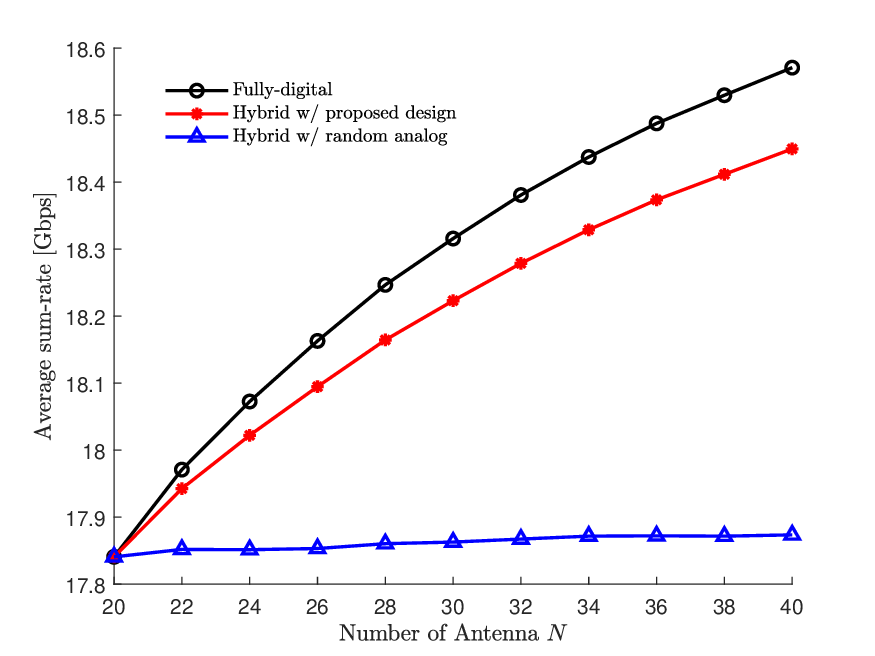}
\caption{\label{fig:avgRsum-vs-antenna}Average sum-rate versus AP antennas $N$ with $M=4$, $L=8$, $C_F=10$, and $K=20$.}
\end{figure}

Lastly, in Fig. \ref{fig:avgRsum-vs-antenna}, we examine the impact of the proposed hybrid analog-digital combining design in Sec. \ref{sub:hybrid-architecture} by plotting the average sum-rate versus the number of antennas $N$ for $M=4$ stripes, $L=8$ APs per stripe, $C_F=10$ Gbps, and $K=20$ UEs. We compare the performance of three schemes: \textit{i)} Proposed fully-digital solution (Fully-digital); \textit{ii)} Proposed hybrid analog-digital solution (Hybrid w/ proposed design); and \textit{iii)} Hybrid analog-digital solution with random analog combining matrices (Hybrid w/ random analog).
In the last scheme, each $(n,k)$th element of analog combining matrix $\mathbf{U}_{m,i}^A$ is given by $e^{j \theta_{m,i,n,k}}$, where the phases $\theta_{m,i,n,k}$ are independent across the indices $m$, $i$, $n$, and $k$, and identically follow a uniform distribution $\mathcal{U}(-\pi, \pi)$.
For the hybrid schemes, the number of RF chains is fixed at $K$, as assumed in Sec. \ref{sub:hybrid-architecture}. 
The figure shows that as the number of AP antennas $N$ increases, the sum-rate of the hybrid scheme with random analog combining matrices remains relatively constant. In contrast, the proposed hybrid analog-digital design leads to a sum-rate that grows with $N$, achieving performance comparable to the fully-digital scheme.
This comparison highlights the effectiveness of the proposed hybrid scheme.

\section{Conclusion} \label{sec:conclusion}
We have investigated the uplink CF-mMIMO system with capacity-constrained parallel radio stripe networks. Specifically, we have focused on INP at APs, where each AP performs linear combining of baseband signals received from the uplink wireless channel and the incoming fronthaul link and a compression of the combining output signals.
To address the complexity and the CSI requirement associated with jointly optimizing the INP strategies across all APs, we have proposed an efficient sequential design on a per-AP basis. Numerical results demonstrated that the proposed design yields significant gains compared to benchmark schemes, with marginal gaps in comparison to the cutset upperbound.
For future work, an extension to a multi-antenna UE case is worth pursuing. In this system, we need to optimize the INP operations and transmit covariance matrices of multi-antenna UEs jointly.
Additionally, we aim to explore the application of more advanced coding techniques, such as joint decompression and decoding (see, e.g., \cite{Lim:TIT11, Park:SPL13}) to further enhance performance and approach the cutset bound.

\section*{Notations}

Uppercase and lowercase boldface symbols represent matrices and vectors, respectively. The complex conjugate, transpose, and Hermitian transpose operators are denoted by $(\cdot)^*$, $(\cdot)^T$, and $(\cdot)^H$, respectively, and the expectation over the random variable $X$ is expressed as $\mathbb{E}_X[\cdot]$. The function $\text{diag}(\cdot)$ creates a diagonal matrix with the elements of the input vector on its main diagonal, while $\text{blkdiag}(\cdot)$ constructs a block diagonal matrix from the input submatrices. $\mathbf{x}\sim\mathcal{CN}(\mathbf{0}, \boldsymbol{\Sigma})$ indicates that the vector $\mathbf{x}$ follows a circularly-symmetric complex Gaussian distribution with zero mean and covariance matrix $\boldsymbol{\Sigma}$. The mutual information between random variables $X$ and $Y$ is denoted by $I(X;Y)$. The notation $\|\cdot\|^2$ represents the Euclidean norm. The set of $m\times n$ complex matrices is denoted by $\mathbb{C}^{m\times n}$, and $\mathbf{X} \succeq \mathbf{0}$ indicates that $\mathbf{X}$ is a positive semidefinite matrix.
We denote the Kronecker delta function by $\delta_{i,j}$, where $\delta_{i,j}=1$ if $i=j$ and $\delta_{i,j}=0$ otherwise.


\appendices

\section{Proof of Theorem 1}

In this appendix, we prove Theorem 1, showing that the optimal solution to problem (\ref{eq:problem-successive-Omega-modified}) is given by (\ref{eq:local-optimum}) with $\lambda_{m,i}$ chosen such that (\ref{eq:condition-optimal-Lagrange-multiplier}) is satisfied.

Note that the Slater's condition is satisfied in (\ref{eq:problem-successive-Omega-modified}), indicating that the strong duality holds. Thus, we can identify an efficient solution to (\ref{eq:problem-successive-Omega-modified}) by leveraging the Lagrange duality method. The Lagrangian of the problem (\ref{eq:problem-successive-Omega-modified}) is given by
\begin{align}
    &\mathcal{L}_{m,i}\left(\mathbf{a}_{m,i}, \lambda_{m,i}, \boldsymbol{\mu}_{m,i}\right) = \nonumber \\
    & \sum\nolimits_{k\in\mathcal{K}} \!\log_2\left( 1 + a_{m,i,k}\left(\gamma_{m,i,k}^{\text{eig}} + 1\right)\right) \nonumber \\
    & - \sum\nolimits_{k\in\mathcal{K}} \log_2\left( 1 + a_{m,i,k}\right) \nonumber \\
    &- \lambda_{m,i}\left[ \sum\nolimits_{k\in\mathcal{K}} \log_2\left( 1 + a_{m,i,k}\left(\gamma_{m,i,k}^{\text{eig}} + 1\right) \right)  - C_F \right] \nonumber \\
    & + \sum\nolimits_{k\in\mathcal{K}} \mu_{m,i,k} a_{m,i,k}, \nonumber
\end{align}
where $\lambda_{m,i}$ and $\boldsymbol{\mu}_{m,i} = \{\mu_{m,i,k}\}_{k\in\mathcal{K}}$ are the Lagrange multipliers associated with the constraints in \eqref{eq:problem-successive-Omega-modified-fronthaul-capacity} and \eqref{eq:problem-successive-Omega-modified-nonnegativity}, respectively.

The Karush-Kuhn-Tucker (KKT) conditions for the problem (\ref{eq:problem-successive-Omega-modified}) can be written as \cite{Boyd:Cambridge}
\begin{subequations}\label{eq:KKT-conditions}
\begin{align} 
    & \frac{\left(1-\lambda_{m,i}\right)\left(\gamma_{m,i,k}^{\text{eig}} + 1\right)}{1 + a_{m,i,k}\left(\gamma_{m,i,k} + 1\right)} - \frac{1}{1 + a_{m,i,k}} \nonumber \\
    & \,\,\,\,\,\,\,\,\,\,+ \mu_{m,i,k}\ln 2 = 0, \label{eq:KKT-1} \\
    & \lambda_{m,i}\left[ \sum\nolimits_{k\in\mathcal{K}} \log_2\left( 1 + a_{m,i,k}\left(\gamma_{m,i,k}^{\text{eig}} + 1\right) \right)  - C_F \right] \nonumber \\
    & \,\,\,\,\,\,\,\,\,\,= 0, \label{eq:KKT-2} \\
    & \mu_{m,i,k} a_{m,i,k} = 0, \, \forall k\in\mathcal{K}, \label{eq:KKT-3}
\end{align}    
\end{subequations}
where (\ref{eq:KKT-1}) represents the stationary condition, and (\ref{eq:KKT-2}) and (\ref{eq:KKT-3}) state the complementary slackness conditions.
For brevity, the primal and dual feasibility conditions are omitted from (\ref{eq:KKT-conditions}).
From (\ref{eq:KKT-1}) and (\ref{eq:KKT-3}), we observe that if $a_{m,i,k} > 0$ for some $k\in\mathcal{K}$, then the equality $\mu_{m,i,k}=0$ should hold, resulting in the expression for $a_{m,i,k}$ as (\ref{eq:local-optimum}).

It is not difficult to see that the objective function of (\ref{eq:problem-successive-Omega-modified}) is non-decreasing with respect to each $a_{m,i,k}$ since its derivative is non-negative for $\gamma_{m,i,k}^{\text{eig}}\geq0$. Since the left-hand-side of \eqref{eq:problem-successive-Omega-modified-fronthaul-capacity} is monotonically increasing with $a_{m,i,k}$, $\forall k\in\mathcal{K}$, at the optimal point, the fronthaul capacity constraint in \eqref{eq:problem-successive-Omega-modified-fronthaul-capacity} always holds with the equality.

\end{document}